\documentclass[oribibl,envcountsame,envcountsect,runningheads]{llncs}

\usepackage[all]{xy}
\SelectTips{cm}{}
\usepackage{amsmath, bbold, stmaryrd, dsfont}

\usepackage{MnSymbol}
\usepackage{enumerate}
\usepackage{url}
\usepackage{xspace}

\newdir{ >}{{}*!/-5pt/\dir{>}}
\newdir{ (}{{}*!/-10pt/\dir^{(}}
\newdir{ )}{{}*!/-10pt/\dir_{(}}

\spnewtheorem{assumptions}[theorem]{Assumptions}{\bfseries}{}
\spnewtheorem{prop}[theorem]{Proposition}{\bfseries}{\it}
\spnewtheorem{cor}[theorem]{Corollary}{\bfseries}{\it}
\spnewtheorem{defn}[theorem]{Definition}{\bfseries}{}
\spnewtheorem{thm}[theorem]{Theorem}{\bfseries}{\it}
\spnewtheorem{lem}[theorem]{Lemma}{\bfseries}{\it}
\spnewtheorem{rem}[theorem]{Remark}{\bfseries}{}
\spnewtheorem{expl}[theorem]{Example}{\bfseries}{}
\spnewtheorem{observation}[theorem]{Observation}{\bfseries}{}
\spnewtheorem{assumption}[theorem]{Assumption}{\bfseries}{}
\spnewtheorem{notation}[theorem]{Notation}{\bfseries}{}
\spnewtheorem{construction}[theorem]{Construction}{\bfseries}{}

\def\hyph{-\penalty0\hskip0pt\relax}
\def\epsilon{\varepsilon}
\def\eps{\epsilon}

\renewcommand{\rho}{\varrho}
\def\ol{\overline}
\def\inj{\mathsf{in}}
\def\o{\cdot}

\newcommand{\Quo}{\mathsf{Quo}}

\newcommand{\takeout}[1]{\empty}

\newcommand{\ra}{\rightarrow}
\newcommand{\xra}{\xrightarrow}

\newcommand{\Ra}{\Rightarrow}
\newcommand{\La}{\Leftarrow}

\newcommand{\seq}{\subseteq}
\newcommand{\hra}{\hookrightarrow}
\newcommand{\tl}{\widetilde}

\newcommand{\DCath}{\hat \DCat}

\newcommand{\Stone}{\mathsf{Stone}}

\newcommand{\Priest}{\mathsf{Priest}}

\newcommand{\Mon}{\mathsf{Mon}}

\newcommand{\Alg}{\mathsf{Alg}\,}
\newcommand{\Set}{\mathsf{Set}}

\newcommand{\BA}{\mathsf{BA}}

\newcommand{\DL}{\mathsf{DL}_{01}}

\newcommand{\Poset}{\mathsf{Pos}}

\newcommand{\ACat}{\mathcal{A}}
\newcommand{\BCat}{\mathcal{B}}
\newcommand{\Cat}{\mathcal{C}}
\newcommand{\ev}{\mathsf{ev}}
\newcommand{\DCat}{\mathcal{D}}

\newcommand{\id}{\mathsf{id}}
\newcommand{\Id}{\mathsf{Id}}

\newcommand{\PS}{\Psi\Sigma^*}
\newcommand{\VPS}{\und{\Psi\Sigma^*}}
\newcommand{\mlhs}{\mu\hL_\Sigma}
\newcommand{\mls}{\mu L_\Sigma}

\newcommand{\AlgCat}[1]{\mathsf{Alg}~#1}

\newcommand{\Quofp}[1]{\mathsf{Quo}_{f}(#1)}

\newcommand{\Coalg}[1]{\mathsf{Coalg}\,#1}
\newcommand{\Bim}{\mathsf{Mon}(\DCat)}
\newcommand{\Bimfp}{\mathsf{Mon}_{f}(\DCat)}
\newcommand{\SBim}{\Sigma\text{-}\mathsf{Mon}(\DCat)}
\newcommand{\SBimfp}{\Sigma\text{-}\mathsf{Mon}_{f}(\DCat)}
\newcommand{\Coalgfp}[1]{\mathsf{Coalg}_{f}\,#1}
\newcommand{\Coalglfp}[1]{\mathsf{Coalg}_{lfp}\,#1}

\newcommand{\Algcfp}[1]{\mathsf{Alg}_{cfp}\, #1}
\newcommand{\Algfp}[1]{\mathsf{Alg}_{f}\, #1}

\newcommand{\FCoalg}[1]{\mathsf{Coalg}_f\,#1}

\newcommand{\FAlg}[1]{\mathsf{Alg}_{f}\,#1}

\newcommand{\lfp}{locally finitely presentable\xspace}
\newcommand{\lcp}{locally cofinitely presentable\xspace}

\newcommand{\Vect}[1]{\mathsf{Vect}\,#1}

\newcommand{\JSL}{{\mathsf{JSL}_0}}

\newcommand{\colim}{\mathop{\mathsf{colim}}}

\newcommand{\Sub}[1]{\mathsf{Sub}(#1)}

\newcommand{\FPSub}[1]{\mathsf{Sub}_{f}(#1)}
\newcommand{\FPSubrqc}[1]{\mathsf{Sub}_{f}^{r}(#1)}
\newcommand{\Subrqc}[1]{\mathsf{Sub}^{r}(#1)}

\newcommand{\und}[1]{#1}
\newcommand{\under}[1]{|#1|}

\newcommand{\xto}[1]{\xrightarrow{#1}}

\newcommand{\epito}{\twoheadrightarrow}

\newcommand{\monoto}{\rightarrowtail}

\newcommand{\hL}{{\hat{L}}}
\newcommand{\hU}{\hat{U}}
\newcommand{\hF}{\hat{F}}

\newcommand{\Pow}{\mathcal{P}}

\newcommand{\FPow}{\Pow_\omega}

\newcommand{\Ind}{\mathsf{Ind}\,}
\newcommand{\Pro}{\mathsf{Pro}\,}
\newcommand{\ProC}[1]{\mathsf{Pro}(#1)}

\newcommand{\IndC}[1]{\mathsf{Ind}(#1)}

\newcommand{\Ideal}[1]{\mathsf{Ideal}(#1)}

\newcommand{\Field}{\mathds{Z}}

\newcommand{\one}{\mathbb{1}}

\newcommand{\two}{\mathbb{2}}

\newcommand{\Reg}{\mathsf{Reg}}

\newcommand{\Lor}{\bigvee}

\usepackage{xcolor}
\usepackage{blindtext}
\usepackage[draft]{fixme} 
\usepackage{times}

\begin{document}
%
%

\FXRegisterAuthor{rm}{arm}{Rob}
\FXRegisterAuthor{sm}{asm}{SM}
\FXRegisterAuthor{hu}{ahu}{HU}

%
%
\titlerunning{Generalized Eilenberg Theorem I: Local Varieties of Languages}
\authorrunning{J.~Ad\'{a}mek, S.~Milius, R.~S.~R.~Myers and H.~Urbat}

\title{Generalized Eilenberg Theorem I:\\ Local Varieties of Languages}

\author{Ji\v{r}\'{\i} Ad\'{a}mek\inst{1}, Stefan Milius\inst{2}, Robert S.~R.~Myers\inst{1} \and Henning Urbat\inst{1}}
\institute{Institut f\"ur Theoretische Informatik\\ Technische Universit\"at Braunschweig, Germany
\and
Lehrstuhl f\"ur Theoretische Informatik\\ Friedrich-Alexander-Universit\"at Erlangen-N\"urnberg, Germany
}
\maketitle

\begin{center}Dedicated to Manuela Sobral.\end{center}

\begin{abstract}
We investigate the duality between algebraic and coalgebraic recognition of languages to derive a generalization of the local version of Eilenberg's theorem. This theorem states that the lattice of all boolean algebras of regular languages over an alphabet $\Sigma$ closed under derivatives is isomorphic to the lattice of all pseudovarieties of $\Sigma$-generated monoids. By applying our method to different categories, we obtain three related results: one, due to Gehrke, Grigorieff and Pin, weakens boolean algebras to distributive lattices, one  weakens them to join-semilattices, and the last one considers vector spaces over $\Field_2$.
\end{abstract}

\section{Introduction}
Regular languages are precisely the behaviours of finite automata. A machine\hyph independent
characterization of regularity is the starting point of algebraic
automata theory (see e.g.~\cite{pin13}): one defines recognition via
preimages of monoid morphisms $f: \Sigma^* \to M$, where $M$ is a finite monoid, and every regular language is recognized in this way by
its syntactic monoid. This suggests to investigate how operations on regular languages relate to operations on monoids.  Recall that a \emph{pseudovariety of monoids} is a class of finite monoids closed under finite products, submonoids and quotients (homomorphic images), and a  \emph{variety of regular languages} is a class of regular languages closed under the boolean operations (union, intersection and complement), left and right derivatives\footnote[1]{The left and right derivatives of a language $L\seq \Sigma^*$ are  $w^{-1}L = \{u\in\Sigma^*: wu\in L\}$ and $Lw^{-1}=\{u\in \Sigma^*: uw\in L\}$ for $w\in \Sigma^*$, respectively.} and preimages of monoid morphisms $\Sigma^*\ra\Gamma^*$.
 Eilenberg's variety theorem \cite{Eilenberg76}, a cornerstone of automata theory, establishes a lattice isomorphism
\[ \text{varieties of regular languages} ~\cong~ \text{pseudovarieties of monoids}.\] Numerous variations of this correspondence are known, e.g. weakening the closure properties in the definition of a variety \cite{pin95,polak01}, or replacing regular languages by formal power series \cite{reut80}. Recently Gehrke, Grigorieff and
Pin~\cite{ggp08,ggp10} proved a ``local''
version of Eilenberg's theorem: for every fixed alphabet
$\Sigma$, there is a lattice isomorphism between \emph{local varieties of regular languages} (sets of regular languages over $\Sigma$ closed under boolean operations and derivatives) and \emph{local pseudovarieties of monoids} (sets of $\Sigma$-generated finite monoids closed under quotients and subdirect products). At the heart of this result lies the use of Stone duality to relate the boolean algebra of regular languages over $\Sigma$, equipped with left and right derivatives, to the free $\Sigma$-generated profinite monoid.

In this paper we generalize the local Eilenberg theorem to the level of an abstract duality. Our approach starts with the observation that all concepts involved in this theorem are inherently categorical:
\begin{enumerate}[(1)]
\item The boolean algebra $\Reg_\Sigma$ of all regular languages over $\Sigma$ naturally carries the structure of a deterministic automaton whose transitions $L\xra{a} a^{-1}L$ for $a\in \Sigma$ are given by left derivation and whose final states are the languages containing the empty word. In other words, $\Reg_\Sigma$ is a coalgebra for the functor $T_\Sigma Q = \two \times Q^\Sigma$ on the category of boolean algebras, where $\two=\{0,1\}$ is the two-chain. The coalgebra $\Reg_\Sigma$ can be captured abstractly as the \emph{rational fixpoint} $\rho T_\Sigma$ of $T_\Sigma$, i.e., the terminal locally finite
$T_\Sigma$-coalgebra \cite{m10}.
\item Monoids are precisely the monoid objects in the category of sets, viewed as a monoidal category w.r.t. the cartesian product.
\item The categories of boolean algebras and sets occurring in (1) and (2)  are locally finite varieties of algebras (that is, their finitely generated algebras are finite), and the full subcategories of  \emph{finite} boolean algebras and  \emph{finite} sets are dually equivalent via Stone duality.
\end{enumerate}
Inspired by (3), we call two locally finite varieties $\Cat$ und $\DCat$ of (possibly ordered) algebras  \emph{predual} if the respective full subcategories $\Cat_f$ and $\DCat_f$ of finite algebras are dually equivalent. Our aim is to prove a local Eilenberg theorem for an abstract pair of predual varieties $\Cat$ and $\DCat$, the classical case being covered by taking $\Cat$ = boolean algebras and $\DCat$ = sets. In this setting deterministic automata are modeled both as coalgebras for the functor
 \[T_\Sigma:\Cat\ra\Cat,\quad T_\Sigma Q = \two \times Q^\Sigma,\] 
and as algebras for the functor
\[L_\Sigma: \DCat\ra\DCat,\quad L_\Sigma A = \one + \coprod_{\Sigma} A,\] 
where $\two$ is a two-element algebra in $\Cat$ and $\one$ is its dual finite algebra in $\DCat$. These functors are \emph{predual} in the sense that their restrictions $T_\Sigma: \Cat_f\ra \Cat_f$ and $L_\Sigma: \DCat_f\ra \DCat_f$ to finite algebras are dual, and therefore the categories of finite $T_\Sigma$-coalgebras and finite $L_\Sigma$-algebras are dually equivalent.  As a first approximation to the local Eilenberg theorem, we consider the rational fixpoint $\rho T_\Sigma$ for $T_\Sigma$ -- which is always carried by the automaton $\Reg_\Sigma$ of regular languages --  and the initial algebra $\mu L_\Sigma$ for $L_\Sigma$ and establish a lattice isomorphism
\[\text{subcoalgebras of $\rho T_\Sigma$ ~$\cong$~ ideal completion of the poset of finite quotient algebras of $\mls$}.\]
 This is ``almost'' the
desired general local Eilenberg theorem. For the classical case ($\Cat=$  boolean algebras and $\DCat$ = sets) one has $\rho T_\Sigma = \Reg_\Sigma$ and $\mu L_\Sigma = \Sigma^*$, and the above isomorphism states that the boolean subalgebras of $\Reg_\Sigma$ closed under \emph{left}
derivatives correspond to sets of finite quotient
algebras of $\Sigma^*$ closed under quotients and
subdirect products. What is missing is the closure under \emph{right} derivatives on the coalgebra side, and quotient algebras of $\Sigma^*$ which are
\emph{monoids} on the algebra side. 

The final step is to prove that the above isomorphism restricts to one between \emph{local varieties of regular languages in $\Cat$} (i.e., subcoalgebras of $\rho T_\Sigma$ closed under right derivatives) and
\emph{local pseudovarieties of $\DCat$-monoids}. Here a \emph{$\DCat$-monoid} is a monoid object in the monoidal category $(\DCat,\otimes,\Psi 1)$ where $\otimes$ is the tensor product of algebras and $\Psi 1$ is the free algebra on one generator. In more elementary terms, a $\DCat$-monoid is an algebra $A$ in $\DCat$ equipped with a ``bilinear'' monoid
multiplication  $A\times A \xra{\circ} A$, which means that the maps $a\circ -$
and $- \circ a$ are $\DCat$-morphisms for all $a\in A$. For example, $\DCat$-monoids in $\DCat=$ sets, posets, join-semilattices and vector spaces over $\Field_2$ are monoids, ordered monoids, idempotent semirings and $\Field_2$-algebras (in the sense of algebras over a field), respectively. In all these examples the monoidal category $(\DCat, \otimes, \Psi 1)$ is \emph{closed}: the set $\DCat(A,B)$ of homomorphisms between two algebras $A$ and $B$ is an algebra in $\DCat$ with the pointwise algebraic structure. Our main result is the \\

\noindent\fcolorbox{lightgray}{lightgray}{\parbox{0.98\textwidth}{%
  \textbf{General Local Eilenberg Theorem.} Let $\Cat$ and $\DCat$ be predual locally finite varieties of algebras, where the algebras in $\DCat$ are possibly ordered. Suppose further that  $\DCat$ is monoidal closed w.r.t. tensor product, epimorphisms in $\DCat$ are surjective, and the free algebra in $\DCat$ on one generator is dual to a two-element algebra in $\Cat$. Then there is a lattice isomorphism
  \[\text{\textbf{local varieties of regular languages in $\boldsymbol{\Cat}$ ~$\boldsymbol{\cong}$~ local pseudovarieties of $\boldsymbol{\DCat}$-monoids.}}\]
}}
{~}\\~\\

By applying this to Stone duality ($\Cat= $ boolean algebras and $\DCat=$ sets) we recover the ``classical'' local Eilenberg theorem.  Birkhoff duality ($\Cat$ = distributive lattices and $\DCat= $ posets) gives another result of Gehrke et. al, namely a lattice isomorphism  between \emph{local lattice varieties of regular languages} (subsets of $\Reg_\Sigma$ closed under union, intersection and derivatives) and local pseudovarieties of ordered monoids. Finally, by taking $\Cat=\DCat= $ join-semilattices and $\Cat=\DCat=$ vector spaces over $\Field_2$, respectively, we obtain two new local Eilenberg theorems. The first one establishes a lattice isomorphism between \emph{local semilattice varieties of regular languages} (subsets of $\Reg_\Sigma$ closed under union and derivatives) and local pseudovarieties of idempotent semirings, and the second one gives an isomorphism between \emph{local linear varieties of regular languages} (subsets of $\Reg_\Sigma$ closed under symmetric difference and derivatives) and local pseudovarieties of $\Field_2$-algebras.

As a consequence of the General Local Eilenberg Theorem we also gain a generalized view of profinite monoids. The dual equivalence between $\Cat_f$ and  $\DCat_f$ lifts to a dual equivalence  between $\Cat$ and a category $\DCath$  arising as a profinite completion of $\DCat_f$. In the classical case we have $\Cat$ = boolean algebras, $\DCat$ = sets and $\DCath$ = Stone spaces, and the dual equivalence between $\Cat$ and $\DCath$ is given by Stone duality. Then the $\DCath$-object dual to the rational fixpoint $\rho T_\Sigma \in \Cat$ can be  equipped with a monoid structure that makes it the free profinite $\DCat$-monoid on $\Sigma$.

\noindent\fcolorbox{lightgray}{lightgray}{\parbox{0.98\textwidth}{%
  \textbf{Theorem.} Under the assumptions of the General Local Eilenberg Theorem, the free profinite $\DCat$-monoid on $\Sigma$ is  dual to the rational fixpoint $\rho T_\Sigma$.
}}
{~}\\~\\

\noindent This extends the corresponding results of Gehrke, Grigorieff and Pin \cite{ggp08} for $\DCat$ = sets and $\DCat$ = posets.\\

\noindent The present paper is a revised and extended version of the conference paper \cite{ammu14}, providing full proofs and more detailed examples. In comparison to \emph{loc. cit.} we work with a slightly modified categorical framework in order to simplify the presentation, see Section \ref{sec:oldasm}.

\paragraph{Related work.} Our paper is inspired by the work of Gehrke, Grigorieff and
Pin~\cite{ggp08,ggp10} who showed that the algebraic operation of the free
profinite monoid on $\Sigma$ dualizes to the derivative operations on
the boolean algebra of regular languages (and similarly for the free
ordered profinite monoid on $\Sigma$). Previously, the duality between
the boolean algebra of regular languages and the Stone space of
profinite words appeared (implicitly) in work by
Almeida~\cite{Almeida94} and
was formulated by Pippenger~\cite{Pippenger97} in terms of Stone duality. 

A categorical approach to the duality theory of regular languages has been suggested by Rhodes and Steinberg \cite{rs2008}. They introduce the notion of a boolean bialgebra, and prove the equivalence of bialgebras and profinite semigroups. The precise connection to their work is yet to be investigated. 

Another related work is Pol\'ak~\cite{polak01} and Reutenauer \cite{reut80}. They
consider what we treat as the example of join-semilattices and vector spaces, respectively, and
obtain a (non-local) Eilenberg theorem for these cases. To the
best of our knowledge the local version we prove does not follow from
the global version, and so we believe that our result is new. 

The origin of all the above work is, of course, Eilenberg's
theorem~\cite{Eilenberg76}. Later Rei\-terman~\cite{Reiterman82} proved another characterization of
pseudovarieties of monoids in the spirit of Birkhoff's classical variety
theorem. Reiterman's theorem states
that any pseudovariety of monoids can be characterized by
profinite equations (i.e., pairs of elements of a free profinite monoid). We do not treat profinite
equations in the present paper.

\section{The Rational Fixpoint}
The aim of this section is to recall the rational fixpoint of a functor, which provides a an abstract coalgebraic view of the set of regular languages. As a prerequisite, we need a categorical notion of ``finite automaton'', and so we will work with categories where enough ``finite'' objects exist -- viz. \emph{locally finitely presentable} categories \cite{ar94}.
\begin{defn}
\label{def:lfp}
\begin{enumerate}[(a)]
\item An object $X$ of a category $\Cat$ is \emph{finitely presentable} if the hom-functor $\Cat(X,-):\Cat\ra\Set$ is finitary (i.e., preserves filtered colimits). Let $\Cat_{f}$ denote the full subcategory of all finitely presentable objects of $\Cat$. 
\item $\Cat$ is \emph{locally finitely presentable} if it is cocomplete, $\Cat_{f}$ is small up to isomorphism and every object of $\Cat$ is a filtered colimit of finitely presentable objects.
\end{enumerate}
\end{defn}

\begin{expl}\label{ex:lfp} Let $\Gamma$ be a finitary signature, that is, a set of operation symbols with finite arity.
\begin{enumerate}[(1)] 
\item Denote by $\Alg{\Gamma}$ the category of $\Gamma$-algebras and $\Gamma$-homomorphisms. A \emph{variety of algebras} is a full subcategory of $\Alg{\Gamma}$ closed under products, subalgebras (represented by injective homomorphisms) and homomorphic images (represented by surjective homomorphisms). Equivalently, by Birkhoff's theorem \cite{Birkhoff35}, varieties of algebras are precisely the classes of algebras definable by equations of the form $s=t$, where $s$ and $t$ are $\Gamma$-terms. Every variety of algebras is locally finitely presentable \cite[Corollary 3.7]{ar94}.
\item Similarly, let $\Alg_{\leq} \Gamma$ be the category of ordered $\Gamma$-algebras. Its objects are $\Gamma$-algebras carrying a poset structure for which every $\Gamma$-operation is monotone, and its morphisms are monotone $\Gamma$-homomorphisms. A \emph{variety of ordered algebras} is a full subcategory of $\Alg_\leq \Gamma$ closed under products, subalgebras and homomorphic images. Here subalgebras are represented by embeddings (injective $\Gamma$-homomorphisms that are both monotone and order-reflecting), and homomorphic images are represented by surjective $\Sigma$-homomorphisms that are monotone but not necessarily order-reflecting. 

Bloom \cite{Bloom1976} proved an ordered analogue of Birkhoff's theorem: varieties of ordered algebras are precisely the classes of ordered algebras definable by inequalities $s\leq t$ between $\Gamma$-terms. From this it is easy to see that every variety of ordered algebras is  finitary monadic over the \lfp category of posets, and hence \lfp \cite[Theorem and Remark 2.78]{ar94}.
\end{enumerate}
 In our applications we will work with the varieties in the table below. Observe  that all these varieties are \emph{locally finite}, that is, their finitely presentable objects are precisely the finite algebras.
\begin{center}
\begin{tabular}{|lll|}
\hline\rule[11pt]{0pt}{0pt}
$\Cat$ & objects & morphisms \\
\hline
$\Set$ & sets & functions\\
$\BA$ & boolean algebras & boolean morphisms\\
$\DL$ & distributive lattices with $0$ and $1\quad$ & lattice morphisms preserving $0$ and $1$\\
$\JSL$ & join-semilattices with $0$  & semilattice morphisms preserving $0$\\
$\Vect{\Field_2}\quad$ & vector spaces over the field $\Field_2$ & linear maps\\
$\Poset$ & partially ordered sets & monotone functions\\
\hline
\end{tabular}
\end{center}
\end{expl}

{~}  

\begin{rem}
For the rest of this paper the term ``variety'' refers to both varieties of algebras and varieties of ordered algebras.
\end{rem}

\begin{notation}
Fix a locally finitely presentable category $\Cat$ and a finitary endofunctor $T: \Cat\ra \Cat$.
\end{notation}

\begin{defn}
\label{nota:fpcoalg}
A \emph{$T$-coalgebra} is a pair $(Q,\gamma)$ of a $\Cat$-object $Q$ and a $\Cat$-morphism $\gamma: Q\ra TQ$. A \emph{homomorphism}
\[ h: (Q,\gamma) \ra (Q',\gamma') \]
of $T$-coalgebras is a $\Cat$-morphism $h: Q\ra Q'$ with $\gamma' \cdot h = Th\cdot \gamma$. We denote by
$\Coalg{T}$ the category of all $T$-coalgebras and their homomorphisms, and by $\FCoalg{T}$ the full subcategory of $T$-coalgebras  $(Q,\gamma)$ with finitely presentable carrier $Q$ (in the case where $\Cat$ is a locally finite variety, these are precisely the finite coalgebras).
\end{defn}

\begin{expl}\label{exp:tcoalg}
Given a finite alphabet $\Sigma$ and an object $\two$ in $\Cat$, the endofunctor
\[ T_\Sigma = \two\times \Id^\Sigma = \two \times \Id\times \Id\times \ldots \times \Id\] of $\Cat$ is finitary since in any \lfp category filtered colimits commute with finite products. If $\Cat$ is a locally finite variety and $\two$ is a two-element algebra in $\Cat$, then $T_\Sigma$-coalgebras are deterministic automata, see e.g. \cite{Rutten2000}. Indeed, by the universal property of the product, to give a coalgebra $Q\xra{\gamma} T_\Sigma Q=\two \times Q^\Sigma$ means precisely to give an algebra $Q$ (of states), morphisms $\gamma_a:Q\ra Q$ for every $a\in\Sigma$ (representing  $a$-transitions) and a morphism $f:Q\ra \two$ (representing final states). Here are two special cases:
\begin{enumerate}[(a)]
\item The usual concept of a deterministic automaton (without initial states) is captured as a coalgebra for $T_\Sigma$ where $\Cat=\Set$ and $\two=\{0,1\}$.
  An important example of a $T_\Sigma$-coalgebra is the automaton $\Reg_\Sigma$ of regular languages. Its states are the regular languages over $\Sigma$, its transitions are
  \[ \gamma_a(L) = a^{-1}L \qquad \text{for all $L\in\Reg_\Sigma$ and  $a\in\Sigma$,} \]
  and the final states are precisely the languages containing the empty word $\epsilon$.
\item Analogously, consider $T_\Sigma$ as an endofunctor of 
   $\Cat=\BA$ with $\two=\{0,1\}$ (the two-element boolean algebra).
   A coalgebra for $T_\Sigma$ is a deterministic automaton with a boolean algebra structure on the state set $Q$. Moreover, the transition maps $\gamma_a:Q\ra Q$ are boolean homomorphisms, and the final states (given by the inverse image of $1$ under $f:Q\ra \two$) form an ultrafilter. The above automaton $\Reg_\Sigma$ is also a $T_\Sigma$-coalgebra in $\BA$: the set of regular languages is a boolean algebra w.r.t. the usual set-theoretic operations, left derivatives preserve these operations, and the languages containing $\epsilon$ form a principal ultrafilter.
\end{enumerate}
\end{expl}

\begin{defn}\label{rem:rat_fix}
\begin{enumerate}[(a)] \item A coalgebra is called \emph{locally finitely presentable} if it is a filtered colimit of coalgebras with finitely presentable carrier.  The full subcategory of $\Coalg{T}$ of all locally finitely presentable coalgebras is denoted $\Coalglfp{T}$. 
\item The \emph{rational fixpoint} of $T$ is the filtered colimit \[r: \rho T\to T(\rho T)\] of \emph{all} coalgebras with finitely presentable carrier, i.e., the colimit of the diagram $\FCoalg{T}\monoto \Coalg{T}$.
\end{enumerate}
\end{defn}
The term ``rational fixpoint'' is justified by item (a) in the theorem below.

\begin{thm}[see \cite{m10}]\label{thm:lfp_indcomp}
\begin{enumerate}[(a)]
\vspace{-0.2cm}\item $r$ is an isomorphism.
\item  $\rho T$ is the terminal locally finitely presentable $T$-coalgebra, i.e., every locally finitely presentable $T$-coalgebra has a unique coalgebra homomorphism into $\rho T$.
\end{enumerate}
\end{thm}

\begin{expl}
\label{exp:ratfix}
The rational fixpoint of $T_\Sigma:\Set\ra\Set$ is the automaton
$\rho T_\Sigma = \Reg_\Sigma$
  of Example \ref{exp:tcoalg}(a), see \cite{amv_atwork}. For any locally finitely presentable $T_\Sigma$-coalgebra $(Q,\gamma)$, the unique homomorphism $(Q,\gamma)\ra \rho T_\Sigma$ maps each state $q\in Q$ to its accepted language
  \[ L_q = \{a_1\ldots a_n \in \Sigma^* ~:~  q \xra{a_1} q_1 \xra{a_2} q_2 \ra \cdots \xra{a_n} q_n \text{ for some final state } q_n\}.\]
  
\end{expl}
This example can be generalized:

\begin{thm}\label{thm:rhotlift}
Suppose that $\Cat$ is a locally finite variety and $T$ lifts a finitary functor $T_0$ on $\Set$, that is, the following diagram (where $U$ denotes the forgetful functor) commutes: 
\[
\xymatrix{
\Cat \ar[r]^T \ar[d]_{U}& \Cat \ar[d]^{U}\\
\Set \ar[r]_{T_0} & \Set
}
\] 
Then the functor $\mathbb{U}: \Coalg{T}\ra \Coalg{T_0}$ given by \[ Q \xra{\gamma} TQ \quad\mapsto\quad UQ \xra{U\gamma} UTQ = T_0 UQ\]
preserves the rational fixpoint, i.e.,
\[\mathbb{U}(\rho T) \cong \rho T_0.\]
\end{thm}

\begin{proof}
The functor $\mathbb{U}$ is finitary (since filtered colimits of $T$-coalgebras are formed on the level of $\Cat$ and hence on the level of $\Set$) and restricts to finite coalgebras, so we have a commutative square
\[
\xymatrix{
\Coalg{T} \ar[r]^{\mathbb{U}} & \Coalg{T_0}\\
\FCoalg{T} \ar@{>->}[u]^I \ar[r]_{\mathbb{V}} & \FCoalg{T_0} \ar@{>->}[u]_{I_0}
}
\]
where $I$ and $I_0$ are the inclusion functors. We will prove below that $\mathbb{V}$ is cofinal, from which the claim follows:
\begin{align*}
\mathbb{U}(\rho T) &= \mathbb{U}(\colim I) & \text{def. } \rho T\\
&\cong \colim(\mathbb{U}I) & \text{$\mathbb{U}$ finitary}\\
&\cong \colim(I_0 \mathbb{V}) & \text{$\mathbb{U}I = I_0 \mathbb{V}$}\\
&\cong \colim(I_0) & \text{$\mathbb{V}$ cofinal}\\
&= \rho T_0 & \text{def. $\rho T_0$}
\end{align*}
The cofinality of $\mathbb{V}$ amounts to proving that
\begin{enumerate}[(1)]
\item for every finite $T_0$-coalgebra $(Q,\gamma)$ there exists a $T_0$-coalgebra homomorphism $(Q,\gamma)\ra \mathbb{V}(Q',\gamma')$ for some finite $T$-coalgebra  $(Q',\gamma')$, and
\item any two such coalgebra homomorphisms are connected by a zig-zag.
\end{enumerate}
Proof of (1). Let $\Phi:\Set\ra\Cat$ be the left adjoint of the forgetful functor $U: \Cat\ra \Set$, and denote the unit and counit of the adjunction by $\eta$ and $\epsilon$, respectively. Given a finite $T_0$-coalgebra $Q\xra{\gamma} T_0 Q$ form the ``free'' $T$-coalgebra 
\[ \Phi Q \xra{\Phi \gamma} \Phi T_0 Q \xra{\Phi T_0 \eta_Q} \Phi T_0 U\Phi Q = \Phi UT \Phi Q \xra{\epsilon_{T\Phi Q}} T\Phi Q.\]
Note that $\Phi Q$ is finite because $\Cat$ is locally finite. Then \[\eta_Q: (Q,\gamma) \ra   \mathbb{V}(\Phi Q, \epsilon_{T\Phi Q}\cdot \Phi T_0\eta_Q \cdot \Phi \gamma)\] is a coalgebra homomorphism. Indeed, the diagram below commutes by the naturality of $\eta$ and the triangle identity $U\epsilon \cdot \eta U = \id$:
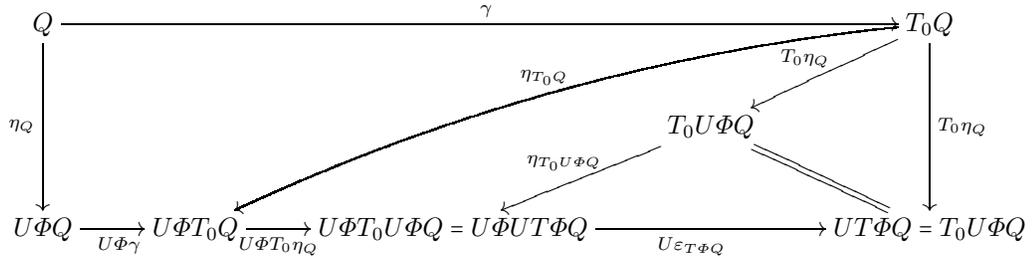
\[
\xymatrix{
Q \ar[rrrr]^\gamma \ar[dd]_{\eta_Q} &&&& T_0 Q \ar[dd]^{T_0\eta_Q} \ar@/_3ex/[ddlll]_{\eta_{T_0 Q}} \ar[dl]_{T_0\eta_Q}\\
&&& T_0U\Phi Q \ar@{=}[dr] \ar[dl]_{\eta_{T_0U\Phi Q}} &\\
U\Phi Q \ar[r]_{U\Phi \gamma} & U\Phi T_0 Q \ar[r]_<<<<<{U\Phi T_0\eta_Q} & U\Phi T_0U\Phi Q = U\Phi UT\Phi Q \ar[rr]_{U\epsilon_{T\Phi Q}}&& UT\Phi Q = T_0 U\Phi Q
}
\]
Proof of (2). Given any coalgebra homomorphism $h: (Q,\gamma)\ra \mathbb{V}(Q',\gamma')$ there exists a unique $\DCat$-morphism $\overline h: \Phi Q \ra Q'$ with $U\overline{h}\cdot \eta_Q = h$ by the universal property of $\eta$. We claim that $\overline{h}$ is a coalgebra homomorphism 
\[\overline{h}: (\Phi Q, \epsilon_{T\Phi Q}\cdot \Phi T_0\eta_Q \cdot \Phi \gamma) \ra (Q',\gamma').\]
Indeed, the lower square in the diagram below commutes when precomposed with $\eta_Q$, from which the equation $\gamma' \cdot \overline{h} = T\overline{h} \circ \epsilon_{T\Phi Q}\cdot \Phi T_0\eta_Q \cdot \Phi \gamma$ follows.
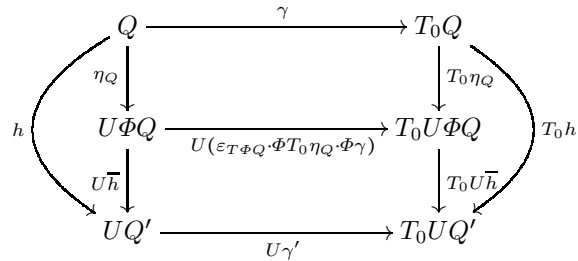
\[
\xymatrix{
Q \ar[rrr]^\gamma \ar[d]_{\eta_Q} \ar@/_8ex/[dd]_h &&& T_0 Q \ar[d]^{T_0\eta_Q} \ar@/^8ex/[dd]^{T_0 h}\\
U\Phi Q \ar[d]_{U\overline{h}} \ar[rrr]_{U(\epsilon_{T\Phi Q}\cdot \Phi T_0\eta_Q \cdot \Phi \gamma)} &&& T_0 U\Phi Q \ar[d]^{T_0U\overline{h}}\\
UQ' \ar[rrr]_{U\gamma'} &&& T_0 UQ'
}
\]
 Now given two coalgebra homomorphisms $h: (Q,\gamma)\ra \mathbb{V}(Q',\gamma')$ and $k: (Q, \gamma) \to \mathbb{V}(Q'', \gamma'')$, the desired zig-zag in $\Coalgfp{T}$ is
$\xymatrix@1@C-.35pc{
Q' & \Phi Q \ar[l]_{\ol h} \ar[r]^-{\ol k} & Q''.
}$\qed
\end{proof}

\begin{cor}\label{cor:ratlift}
Let $\Cat$ be a locally finite variety with a two-element algebra $\two$. Then the rational fixpoint of $T_\Sigma = \two \times \Id^\Sigma: \Cat\ra \Cat$ is carried by the automaton $\Reg_\Sigma$ of Example \ref{exp:tcoalg}(a).  For any locally finitely presentable $T_\Sigma$-coalgebra $(Q,\gamma)$, the unique homomorphism $(Q,\gamma)\ra \rho T_\Sigma$ maps each state $q\in Q$ to its accepted language.
\end{cor}

\begin{proof}
Apply Theorem \ref{thm:rhotlift} to $T=T_\Sigma$ and $T_0 = \{0,1\}\times \Id^\Sigma$. Since $\rho T_0 = \Reg_\Sigma$ by Example \ref{exp:ratfix}, the claim follows. \qed
\end{proof}

Next we will show that the locally finitely presentable $T$-coalgebras arise as a ``free completion'' of the coalgebras with finitely presentable carrier (Theorem \ref{thm:lfpindcomp} below). 

\begin{rem}\label{rem:indcomp}
\begin{enumerate}[(a)]\item Recall that the \emph{free completion under filtered colimits} of a small category $\ACat$ is a full embedding $\ACat\hookrightarrow \Ind{\ACat}$ such that $\Ind{\ACat}$ has filtered colimits and every functor $F: \ACat \ra \BCat$ into a category $\BCat$ with filtered colimits has a finitary extension $\overline{F}: \Ind{\ACat}\ra \BCat$, unique up to natural isomorphim:
\[
\xymatrix{
\ACat \ar@{>->}[r] \ar[dr]_F & \Ind{\ACat} \ar@{-->}[d]^{\overline{F}} \\
& \BCat
}
\]
 This determines $\Ind{\ACat}$ up to equivalence. If $\ACat$ has finite colimits then $\Ind{\ACat}$ is \lfp and $(\Ind{\ACat})_{f} \cong \ACat$. Conversely, every \lfp category $\Cat$ arises in this way: $\Cat\cong\IndC{\Cat_{f}}$.
\item If $\ACat$ is a join-semilattice then $\Ind{\ACat}$ is its ideal completion, see Remark \ref{rem:idcomp}.
\end{enumerate}
\end{rem}

\begin{lem}
  \label{lem:coref}
Let $\BCat$ be a cocomplete category and $J : \ACat \monoto \BCat$ be a small full subcategory of finitely presentable objects closed under finite colimits. Then the unique finitary extension $J^* : \IndC{\ACat} \to \BCat$ forms a full coreflective subcategory.
\end{lem}

\begin{proof}
  By the theorem in Section~VI.1.8 of Johnstone \cite{j82}, we
  know that $J^*$ is a full embedding so that $\IndC\ACat$ can be
  identified with the full subcategory of $\BCat$ given by all
  filtered colimits of objects from $\ACat$. We will show that this
  full subcategory is coreflective. Let $B$ be an object of $\BCat$
  and define $\ol B$ to be the colimit of the filtered diagram
  \[
  \xymatrix@1{
    \ACat /B \ar[r] & \ACat \ar@{ (->}[r] & \BCat,
  }
  \]
  where the first arrow is the canonical projection functor and the
  second one the inclusion functor. We denote the corresponding colimit injections 
  by 
  \[
  \inj_f: A \to \ol B \qquad \text{for every $f: A\to B$ in $\ACat /B$.}
  \]
  Clearly, the objects in $\ACat/B$ form a cocone on the above diagram
  and so we have a unique morphism $b: \ol B \to B$ such that
  \[
  b \cdot \inj_f = f \qquad \text{for every $f: A \to B$ in $\ACat/B$.}
  \]
  We will now prove this morphism $b$ to be couniversal. To this end, let
  $A$ be an object of $\IndC\ACat$, i.\,e.,
  \[
  A = \colim\limits_{i \in I} A_i
  \]
  is a filtered colimit in $\BCat$ of objects from $\ACat$ with
  colimit injections $a_i: A_i \to A$, $i \in I$. Given a morphism $f:
  A \to B$ in $\BCat$, the morphism $f \cdot a_i: A_i \to B$ is an
  object of $\ACat/B$, and for each connecting morphism $a_{i,j}: A_i
  \to A_j$ we have
  \[
  (f \cdot a_j) \cdot a_{i,j} = f \cdot a_i.
  \]
  Thus, $a_{i,j}$ is a morphism in $\ACat /B$ and so 
  \[
  \inj_{f \cdot a_j} \cdot a_{i,j} = \inj_{f\cdot a_i},
  \]
  i.\,e., the morphisms $\inj_{f \cdot a_i}: A_i \to \ol B$ form a
  cocone. So we get a unique $\ol f: A \to \ol B$ such that
  \[
  \ol f \cdot a_i = \inj_{f \cdot a_i} \qquad \text{for every $i \in I$.}
  \]
  Now the following diagram commutes:
  \[
  \xymatrix{
    &
    \ol B
    \ar[r]^-b & B
    \\
    A_i
    \ar[ru]^-{\inj_{f \cdot a_i}}
    \ar[r]_-{a_i}
    &
    A
    \ar[u]_{\ol f}
    \ar[ru]_f
    }
  \]
  Indeed, the outside and left-hand triangle commute, and so the
  right-hand one commutes when precomposed with every $a_i$, $i \in
  I$, whence that triangle commutes since the colimit injections $a_i$
  form a jointly epimorphic family. 

  We still need to show that $\ol f$ is unique such that $b \cdot \ol f =
  f$. So assume that $\ol f: A \to \ol B$ is any such morphism. Fix $i
  \in I$. Then, since $\ol B$ is a filtered colimit and $A_i$ is
  finitely presentable, it follows that there exists some $g: A_i' \to
  B$ in $\ACat/B$ and some morphism $\ol f': A_i \to A_i'$ such that
  the square below commutes:
  \[
  \xymatrix{
    A_i' 
    \ar[r]^-{\inj_g} 
    & 
    \ol B
    \\
    A_i
    \ar[u]^{\ol f '}
    \ar[r]_-{a_i}
    &
    A
    \ar[u]_{\ol f}
    }
  \]
  It follows that $\ol f'$ is a connecting morphism in $\ACat/B$ from
  $f \cdot a_i$ to $g$:
  \[
  g \cdot \ol f' = b \cdot \inj_g \cdot \ol f' = b \cdot \ol f \cdot a_i = f \cdot a_i.
  \]
  Therefore we get $\inj_g \cdot \ol f' = \inj_{f \cdot a_i}$ so that
  \[
  \ol f \cdot a_i = \inj_g \cdot \ol f' = \inj_{f \cdot a_i}, 
  \]
  which determines $\ol f$ uniquely. This completes the proof. \qed
\end{proof}

\begin{thm}\label{thm:lfpindcomp}
 $\Coalglfp{T}$ is the $\Ind$-completion of $\FCoalg{T}$ and forms a coreflective subcategory of $\Coalg{T}$.
\end{thm}

\begin{proof}
We apply the previous lemma to $\ACat = \Coalgfp{T}$ and $\BCat = \Coalg{T}$. Then $\BCat$ is clearly cocomplete, and $\ACat$ is closed under finite colimits (since colimits of coalgebras are constructed in the base category and finitely presentable objects are closed under finite colimits). Moreover, as shown in ~\cite{ap04}, every $T$-coalgebra with finitely presentable carrier is a  finitely presentable object of $\Coalg{T}$.

Hence Lemma~\ref{lem:coref} yields that $J^*: \IndC{\Coalgfp{T}}\hra \Coalg{T}$ is a full coreflective subcategory. The definition of $J^*$ is that it takes formal filtered diagrams of objects in $\Coalgfp{T}$ and constructs their colimits. Therefore its image is precisely $\Coalglfp{T}$, which gives the desired equivalence $\Coalglfp{T} \cong \IndC{\FCoalg{T}}$ and that $\Coalglfp{T}$ is coreflective.\qed
\end{proof}

\section{Algebraic and Coalgebraic Recognition}\label{sec:algcoalgrec}

We are ready to present our first take on the local Eilenberg theorem. At the heart of our approach lies the investigation of a duality for our categories of interest (e.g. Stone duality between finite boolean algebras and finite sets) and the induced algebra-coalgebra duality.

\begin{defn}
Two categories $\Cat$ and $\DCat$ are called \emph{predual} if their full subcategories $\Cat_f$ and $\DCat_f$ of finitely presentable objects are dually equivalent, that is, $\Cat_f \cong \DCat_f^{op}$.
\end{defn}

\begin{expl}\label{expl:dualpairs}
The pairs of locally finite varieties listed in the table below are predual.
\begin{center}
\begin{tabular}{|ll|}
\hline\rule[11pt]{0pt}{0pt}
$\Cat\quad\quad\quad$   & $\DCat\quad\quad\quad$ \\
\hline
$\BA$  & $\Set$\\ 
$\DL$ & $\Poset$\\ 
$\JSL$  & $\JSL$\\
$\Vect{\Field_2}\quad$  & $\Vect{\Field_2}$\\
\hline
\end{tabular}
\end{center}
In more detail:
\begin{enumerate}[(a)]
  \item The categories $\BA$ and $\Set$ are predual via Stone duality. The equivalence $\BA_f \xra{\cong} \Set_f^{op}$ assigns to each finite boolean algebra the set of all atoms, and its associated equivalence $\Set_f^{op} \xra{\cong} \BA_f$ sends each finite set to the boolean algebra of all subsets.
\item The categories $\DL$ and $\Poset$ are predual via Birkhoff duality. The equivalence $\mathsf{DL}_{01,f} \xra{\cong} \Poset_f^{op}$ assigns to each finite distributive lattice the subposet of all join-irreducible elements, and its associated equivalence $\Poset_f^{op} \xra{\cong} \mathsf{DL}_{01,f}$ sends each finite poset to the lattice of all down-closed subsets.
\item The category $\JSL$ is self-predual. The equivalence $\mathsf{JSL}_{0,f} \xra{\cong} \mathsf{JSL}_{0,f}^{op}$ sends each finite join-semilattice $X$ to its dual poset $X^{op}$.
\item The category $\Vect{\Field_2}$ is self-predual. The equivalence $\Vect_f{\Field_2} \xra{\cong} (\Vect_f{\Field_2})^{op}$ sends each finite $\Field_2$-vector space $X$ to the dual space $X^* = \hom(X,\Field_2)$ of all linear maps from $X$ to $\Field_2$.
\end{enumerate}
\end{expl}

\begin{defn}
Let $\Cat$ and $\DCat$ be predual categories. Two functors $T: \Cat \ra \Cat$ and $L: \DCat \ra \DCat$ are called \emph{predual} if they restrict to functors $T_f: \Cat_f\ra\Cat_f$ and $L_f: \DCat_f\ra\DCat_f$ and these restrictions are dual, i.e., the following diagram commutes up to natural isomorphism:
\[
\xymatrix{
\DCat_f^{op} \ar[r]^{L_f^{op}}  & \DCat_f^{op}  \\
\Cat_f \ar[u]^{\cong} \ar[r]_{T_f} & \Cat_f \ar[u]_{\cong}  
}
\]
\end{defn}

\begin{assumptions}\label{asm:sec3}
For the remainder of this paper we fix the following data:
\begin{enumerate}[(a)]
\item $\Cat$ is a locally finite variety of algebras and $\DCat$ is a locally finite variety of algebras or ordered algebras. 
\item $\Cat$ and $\DCat$ are predual with equivalence functors
\[ \widehat{(\mathord{-})} : \Cat_f \xra{\cong} \DCat_f^{op} \quad\text{and}\quad \overline{(\mathord{-})}: \DCat_f^{op} \xra{\cong} \Cat_f.\]
\item $T: \Cat\ra \Cat$ and $L:\DCat\ra\DCat$ are predual finitary functors. Moreover, $T$ preserves monomorphisms and intersections and $L$ preserves epimorphisms.
\end{enumerate}
\end{assumptions}

\begin{expl}\label{exp:TL}
The endofunctor $T_\Sigma = \two\times \Id^\Sigma$ of $\Cat$ (see Example \ref{exp:tcoalg}) has the predual endofunctor \[L_\Sigma = \one + \coprod_\Sigma \Id\] of $\DCat$, where $\one = \widehat{\two}$. Clearly both functors are finitary, $T_\Sigma$ preserves monomorphisms and intersections and $L_\Sigma$ preserves epimorphisms.
\end{expl}

\begin{defn}
\label{not:lalg}
An \emph{$L$-algebra} $(A,\alpha)$ consists of a $\DCat$-object $A$ and a $\DCat$-morphism $\alpha:LA\ra A$. A \emph{homomorphism}
\[ h: (A,\alpha) \ra (A',\alpha') \]
of $L$-algebras is a $\DCat$-morphism $h: A\ra A'$ with $h\cdot \alpha = \alpha'\cdot Lh$.
We denote by $\Alg{L}$ the category of $L$-algebras, and by $\FAlg{L}$ the full subcategory of finite $L$-algebras, that is, algebras $(A,\alpha)$ with finite carrier $A$.
\end{defn}

\begin{expl} Algebras for the functor $L_\Sigma$ of Example \ref{exp:TL} correspond to deterministic automata in $\DCat$ with an initial state but without final states. Indeed, by the universal property of the coproduct an $L_\Sigma$-algebra $L_\Sigma A = \one + \coprod_\Sigma A \xra{\alpha} A$ is determined by a morphism $i: \one \ra A$ (specifying an initial state) and morphisms $\alpha_a: A\ra A$ for each $a\in \Sigma$ (specifying the $a$-transitions). Here are two special cases:
\begin{enumerate}[(a)]
  \item If $\Cat=\BA$ and $\DCat = \Set$ then $\one=\widehat{\two}$ is the one-element set since the two-element boolean algebra $\two$ has one atom, so $L_\Sigma$-algebras are precisely the (classical) deterministic automata with an initial state but without final states.
 \item Similarly, if $\Cat=\DL$ and $\DCat = \Poset$ then $\one$ is the one-element poset, and $L_\Sigma$-algebras correspond to \emph{ordered} deterministic automata with an initial state but without final states.
\end{enumerate}
\end{expl}

\begin{rem}\label{rem:algcoalgdual}
The categories $\FCoalg{T}$ and $\Algfp{L}$ are dually equivalent. Indeed, the equivalence functor $\widehat{(\mathord{-})}: \Cat_f \xra{\cong} \DCat_f^{op}$ lifts to an equivalence $\FCoalg{T}\xra{\cong} (\Algfp{L})^{op}$ given by 
\[ (Q\xra{\gamma} TQ) \quad\mapsto\quad  (L\widehat{Q} = \widehat{TQ} \xra{\widehat \gamma} \widehat Q). \]
\end{rem}

\begin{notation}\label{not:order}
\begin{enumerate}
  \item By a \emph{subcoalgebra} of a $T$-coalgebra $(Q,\gamma)$ is meant one represented by a homomorphism $m:(Q',\gamma')\monoto(Q,\gamma)$ with $m$ monic in $\Cat$. Subcoalgebras are ordered as usual: $m\leq m'$ iff $m$  factorizes through $m'$ in $\Coalg{T}$. We denote by $\Sub{\rho T}$
the poset of all subcoalgebras of $\rho T$, and by
$\FPSub{\rho T}$
the subposet of all \emph{finite subcoalgebras} of $\rho T$.
\item
Likewise, a \emph{quotient algebra} of an $L$-algebra is one represented by an epimorphism in $\DCat$. Quotient algebras are ordered by $e\leq e'$ iff $e$  factorizes through $e'$ in $\Alg{L}$. For the initial $L$-algebra $\mu L$, which exists because $L$ is finitary, we have the posets $\Quofp{\mu  L}\seq \Quo(\mu L)$ of all (finite) quotient algebras of $\mu L$.
\end{enumerate}
\end{notation}

\begin{rem}\label{rem:algfact}
\begin{enumerate}[(a)]
\item Since $L$ preserves epimorphisms, $\Alg{ L}$ has a factorization system consisting of homomorphisms carried by epimorphisms and strong monomorphisms in $\DCat$, respectively. Indeed,
  given any homomorphism $h: (A, \alpha) \to (B, \beta)$ of
  $ L$-algebras take its (epi, strong mono)-factorization $h = m \cdot e$ in $\DCat$
  and then use that $ L$ preserves epis and diagonalization to obtain an
  algebra structure on the domain of $m$ such that
  $m$ and $e$ are $ L$-algebra homomorphisms:
  \[
  \xymatrix{
     L A
    \ar@{->>}[d]_{ L e}
    \ar[r]^-\alpha
    &
    A
    \ar@{->>}[d]^e
    \\
     L C
    \ar[d]_{ L m}
    \ar@{-->}[r]
    &
    C
    \ar@{ >->}[d]^m
    \\
     L B
    \ar[r]_-{\beta}
    &
    B
    }
  \]
  Moreover, given a commutative square of $L$-algebra homomorphisms, where $e$ is an epimorphism and $m$ is a strong monomorphism in $\DCat$, the unique diagonal $d$ is easily seen to be a homomorphism of $L$-algebras.
  \[ 
 \xymatrix{
 (A,\alpha) \ar@{->>}[r]^e \ar[d]_f & (B,\beta)\ar[d]^g \ar[dl]^d \\
 (C,\gamma) \ar@{>->}[r]_m & (D,\gamma)
 }
   \]
\item Dually, since $T$ preserves monomorphisms, $\Coalg{T}$ has a factorization system of homomorphisms carried by strong epimorphisms and monomorphisms in $\Cat$.
\end{enumerate}
\end{rem}
Using factorizations, locally finitely presentable coalgebras can be described in terms of subcoalgebras.

\begin{prop}\label{prop:locfin}
\begin{enumerate}[(a)]
\item A $T$-coalgebra is locally finitely presentable iff it is locally finite, i.e., every state is contained in a some finite subcoalgebra.
\item Every subcoalgebra of a locally finite coalgebra is locally finite.
\end{enumerate}
\end{prop}

\begin{proof}
(a) Let $(Q,\gamma)$ be a locally finitely presentable coalgebra, i.e., it arises as a filtered $(Q_i,\gamma_i)\xra{c_i} (Q,\gamma)$ ($i\in I$) of finite coalgebras $(Q_i,\gamma_i)$. Since filtered colimits of $T$-coalgebras are formed on the level of $\Cat$ and hence on the level of $\Set$, the maps $c_i$ are jointly surjective. It follows that every state $q\in Q$ is contained in $c_i[Q_i]$ for some $i$, and hence in the subcoalgebra of $(Q,\gamma)$ obtained by factorizing $c_i$ as in Remark \ref{rem:algfact}(b).

Conversely, suppose that every state  is contained in some finite subcoalgebra of $(Q,\gamma)$. Then the filtered cocone $(Q_i,\gamma_i)\monoto (Q,\gamma)$ ($i\in I$) of all finite subcoalgebras of $(Q,\gamma)$ is jointly surjective. This implies that $Q_i \monoto Q$ ($i\in I$) is a filtered colimit in $\Set$ and hence also in $\Cat$ and $\Coalg{T}$.

(b) Let $Q$ be a subcoalgebra of a locally finite coalgebra $Q'$. Then every state $q\in Q$ is contained in some finite subcoalgebra $Q''$ of $Q'$. Since $T$ preserves intersections, $Q''\cap Q$ is a finite subcoalgebra of $Q$ containing $q$, so $Q$ is locally finite.\qed
\end{proof}

\begin{expl}
A coalgebra for the functor $T_\Sigma Q= \two \times Q^\Sigma$ on $\Cat$ is locally finitely presentable iff from every state only finitely many states are reachable by transitions. 
\end{expl}

\begin{prop}\label{prop:subquolat}
$\Sub{\rho T}$ and $\Quo(\mu L)$ are complete lattices, and $\FPSub{\rho T}$
and $\Quofp{\mu  L}$ are join-subsemilattices.
\end{prop}

\begin{proof}
Given a family of subcoalgebras $m_i: (Q_i,\gamma_i)\monoto \rho T$ ($i\in I$) it is easy to see that their join in $\Sub{\rho T}$ is the subcoalgebra $m: (Q,\gamma) \monoto \rho T$ obtained by factorizing the homomorphism $[m_i]: \coprod_i (Q_i,\gamma_i) \ra \rho T$ as in Remark \ref{rem:algfact}(b).
\[
\xymatrix{
\coprod_i (Q_i,\gamma_i) \ar[rr]^{[m_i]} \ar@{->>}[dr] && \rho T\\
& (Q,\gamma) \ar@{>->}[ur]_m  & 
}
\]
If $I$ and all $(Q_i,\gamma_i)$ are finite, then $(Q,\gamma)$ is clearly also finite. This proves that $\Sub{\rho T}$ is a complete lattice and $\FPSub{\rho T}$ is a join-subsemilattice. The corresponding statements about $\Quo(\mu L)$ and $\Quofp{\mu  L}$ are shown by a dual argument.\qed
\end{proof}

\begin{prop}
\label{prop:subquo}
The semilattices $\FPSub{\rho T}$ and $\Quofp{\mu  L}$ are isomorphic. The isomorphism is given by
\[ (m: (Q,\gamma)\monoto \rho T) \quad\mapsto\quad (e: \mu L \epito (\widehat Q, \widehat \gamma))\]
where $e$ is the unique $L$-algebra homomorphism defined by the initiality of $\mu L$.
\end{prop}

\begin{proof}
Consider any finite $T$-coalgebra $(Q,\gamma)$ and its dual finite $L$-algebra $(\widehat Q, \widehat \gamma)$, see Remark \ref{rem:algcoalgdual}. Since $\rho T$ is the terminal locally finite $T$-coalgebra and $\mu L$ is the initial $L$-algebra, there are unique homomorphisms
\[(Q,\gamma)\xra{m} \rho T\quad\text{and}\quad \mu  L \xra{e} (\widehat Q, \widehat \gamma).\]
We will prove that
\[
\text{$m$ is monic (in $\Cat$)} 
\quad \text{iff} \quad
\text{$e$ is epic (in $\DCat$).}
\]
from which the claim immediately follows. Assume first that $m$ is monic in $\Cat$, and let $e=e_2\cdot e_1$ be the factorization of $e$ as in Remark \ref{rem:algfact}(a):
\[
\xymatrix{
\mu  L \ar[rr]^e \ar@{->>}[dr]_{e_1} &&  (\widehat Q,\widehat \gamma)\\
& (A,\alpha) \ar@{>->}[ur]_{e_2} &
}
\]
Since $e_2$ is injective, the algebra $(A,\alpha)$ is finite. Moreover, the strong $\DCat$-monomorphism $e_2$ is also strongly monic in $\DCat_{f}$ because the full embedding $\DCat_{f}\hookrightarrow \DCat$ preserves epis (since $\DCat_{f}$ is closed under finite colimits). Hence the dual morphism $\overline{e_2}$ is strongly epic in $\Cat_{f}$. Since $\rho T$ is the terminal locally finite $T$-coalgebra and $(\ol{A},\ol{\alpha})$ is finite, there exists a unique coalgebra homomorphism $f$ making the triangle below commute:
\[
\xymatrix{
(\overline{\widehat{Q}},\overline{\widehat{\gamma}}) \cong (Q,\gamma) \ar@{>->}[r]^>>>>>m \ar@{->>}[d]_{\overline{e_2}} & \rho T\\
(\overline A,\overline \alpha) \ar[ur]_f&
}
\]
  By assumption $m$ is monic, so $\overline{e_2}$ is monic in $\Cat$ and hence in $\Cat_{f}$. But $\overline{e_2}$ is also strongly epic in $\Cat_{f}$, and hence an isomorphism (both in $\Cat_f$ and $\Cat$). It follows that $e_2$ is an isomorphism (both in $\DCat_f$ and $\DCat$), so $e=e_2\cdot e_1$ is epic in $\DCat$.   

\noindent The converse direction is proved by a symmetric argument.\qed
\end{proof}

\begin{rem}\label{rem:idcomp}
Recall that the \emph{ideal completion} $\Ideal{A}$ of a join-semilattice $A$ is the complete lattice of all ideals (= join-closed downsets) of $A$ ordered by inclusion. Up to isomorphism $\Ideal{A}$ is characterized as a complete lattice containing $A$ such that:
\begin{enumerate}[(1)]
\item every element of $\Ideal{A}$ is a directed join of elements of $A$, and
\item the elements of $A$ are compact in $\Ideal{A}$: if $x\in A$ lies under a directed join of elements $y_i\in \Ideal{A}$, then $x\leq y_i$ for some $i$.
\end{enumerate}
\end{rem}

\begin{thm}\label{thm:latiso} 
$\Sub{\rho T}$ is the ideal completion of $\Quofp{\mu  L}$.
\end{thm}

\begin{proof}
Since $\FPSub{\rho T} \cong \Quofp{\mu L}$ by Proposition \ref{prop:subquo}, it suffices to prove that $\Sub{\rho T}$ (which forms a complete lattice by Proposition \ref{prop:subquolat}) is the ideal completion of $\FPSub{\mu L}$. To this end we verify the properties (1) and (2) of Remark \ref{rem:idcomp}.

(1) We need to prove that every subcoalgebra $m: (Q,\gamma)\monoto \rho T$ is a directed join of subcoalgebras in $\FPSub{\rho T}$. The coalgebra $(Q,\gamma)$ is locally finite, being a subcoalgebra of the locally finite coalgebra $\rho T$ (see Proposition \ref{prop:locfin}), and hence a filtered colimit $c_i: (Q_i,\gamma_i)\ra(Q,\gamma)$ ($i\in I$) of coalgebras in $\Coalgfp{T}$. Factorize each $c_i$ as in Remark \ref{rem:algfact}(b):
\[
\xymatrix{
(Q_i,\gamma_i) \ar@{->>}[dr]_{e_i} \ar[rr]^{c_i} && (Q,\gamma) \ar@{>->}[r]^m & \rho T\\
& (Q_i', \gamma_i') \ar@{>->}[ur]_{m_i} &
}
\]
Then $n_i=m\cdot m_i$ ($i\in I$) is a directed set in $\FPSub{\rho T}$. Since $(c_i)$ is colimit cocone, the morphism $[c_i]: \coprod_{i\in I} Q_i\ra Q$ is a strong epimorphism in $\Cat$. This implies that $\bigcup m_i = id_Q$, hence $\bigcup n_i = m$.

(2) We show that every finite subcoalgebra $m: (Q,\gamma)\monoto \rho T$ is compact in $\Sub{\rho T}$. So let $n=\bigcup_{i\in I} n_i$ be a directed union of subcoalgebras $n_i: (Q_i,\gamma_i)\hra \rho T$ in $\Sub{\rho T}$. Then in $\Cat$ we also have a directed union $n=\bigcup n_i$ because coproducts in $\Coalg{T}$ are formed on the level of $\Cat$. Indeed, $\bigcup n_i$ is formed from $[n_i]$ by an image factorization, see the proof of Proposition \ref{prop:subquolat}, and recall from Remark \ref{rem:algfact} that image factorizations in $\Coalg{T}$ are the liftings of (strong epi, mono)-factorizations in $\Cat$.

Since $Q$ is finite, from $m\seq n$ is follows that $m\seq n_i$ for some $i\in I$. That is, there exists a morphism $f: Q\ra Q_i$ in $\Cat$ with $m=n_i\cdot f$. It remains to verify that $f$ is a $T$-coalgebra homomorphism:
\[
\xymatrix{
Q \ar[r]^\gamma \ar@{>->}[d]_f & TQ \ar@{>->}[d]^{Tf} \\
Q_i \ar@{>->}[d]_{n_i} \ar[r]^{\gamma_i} & TQ_i \ar@{>->}[d]^{Tn_i} \\
\rho T \ar[r]_{r} & T(\rho T)
}
\]
Since the lower square and the outer square commute, it follows that the upper square commutes when composed with $Tn_i$. By assumption $T$ preserves monomorphisms, so we can conclude that upper square commutes.\qed
\end{proof}

\begin{expl}
Let $T_\Sigma: \BA \ra \BA$ and $L_\Sigma: \Set\ra\Set$ as in Example \ref{exp:TL}. The rational fixpoint of $T_\Sigma$ is the boolean algebra $\Reg_\Sigma$ with transitions $L \xra{a} a^{-1}L$, see Corollary \ref{cor:ratlift}, and the initial algebra of $L_\Sigma$ is the automaton $\Sigma^*$ with initial state $\epsilon$ and transitions $w \xra{a} wa$ for $a\in \Sigma$. Hence the previous theorem gives a one-to-one correspondence between
\begin{enumerate}[(i)]
\item sets of regular languages over $\Sigma$ closed under boolean operations and left derivatives, and
\item ideals of $\Quofp{\Sigma^*}$, i.e., sets of quotient automata of $\Sigma^*$ closed under quotients and joins.
\end{enumerate}  
This correspondence is refined in the following section.
\end{expl}

\section{The Local Eilenberg Theorem}\label{sec:locvar}

In this section we establish our main result, the generalized local Eilenberg theorem. We continue to work under the Assumptions \ref{asm:sec3}  -- that is, a locally finite variety of algebras $\Cat$ and a predual locally finite variety of (possibly ordered) algebras $\DCat$ are given -- and restrict our attention to deterministic automata, modeled as coalgebras and algebras for the predual functors 
 \[ T_\Sigma=\two\times \Id^\Sigma:\Cat\ra\Cat \quad\text{and}\quad  L_\Sigma = \one + \coprod_\Sigma \Id:\DCat\ra\DCat,\] 
respectively. Here $\two=\{0,1\}$ is a fixed two-element algebra in $\Cat$ and $\one = \widehat{\two}$ is its dual algebra in $\DCat$.
 
 The crucial step towards Eilenberg's theorem is to prove that the isomorphism 
\[ \FPSub{\rho T_\Sigma} \cong \Quofp{\mu L_\Sigma} \]
of Proposition \ref{prop:subquo} restricts to one between the finite subcoalgebras of $\rho T_\Sigma$ closed under right derivatives and the finite quotient algebras of $\mu L_\Sigma$ whose transitions are induced by a monoid structure. To this end we will characterize right derivatives and monoids from a categorical perspective and show that they are dual concepts (Sections \ref{sec:rightder} and \ref{sec:monoids}). The general local Eilenberg theorem is proved in Section \ref{sec:proofloceil}.

\subsection{Right derivatives}\label{sec:rightder}

The closure of the regular languages under right derivatives is usually proved via the following automata construction: suppose a deterministic $\Sigma$-automaton in $\Set$ (with states $Q$ and final states $F\seq Q$) accepts a language $L\seq \Sigma^*$. Then given $w\in\Sigma^*$ replace the set of final states by
\[ F' = \{q\in Q: q \xra{w} q' \text{ for some } q'\in F \}.\]
The resulting automaton accepts the right derivative $Lw^{-1} = \{u\in \Sigma^*: uw\in L\}$ of $L$. This construction generalizes to arbitrary $T_\Sigma$-coalgebras:

\begin{notation}\label{not:rqc} $T_\Sigma$-coalgebras $Q\xra{\gamma} \two\times Q^\Sigma$ are represented as triples 
\[(Q,\gamma_a: Q\ra Q, f:Q\ra\two),\]
see Example \ref{exp:tcoalg}. For each $T_\Sigma$-coalgebra $Q=(Q,\gamma_a,f)$ and $w\in\Sigma^*$ we put
\[ Q_w := (Q,\gamma_a, f\cdot \gamma_w)\]
where, as usual, $\gamma_w = \gamma_{a_n}\cdot \cdots \cdot \gamma_{a_1}$ for $w=a_1\cdots a_n$.
\end{notation}

\begin{rem}\label{rem:rqcmor}
A $\Cat$-morphism $h: Q\ra Q'$  is a $T_\Sigma$-coalgebra homomorphism $h: (Q,\gamma_a,f)\ra (Q',\gamma_a',f')$ iff the following diagram commutes for all $a\in\Sigma$:
\[
\xymatrix{
Q \ar[r]^{\gamma_a} \ar[d]_h & Q \ar[d]^h \ar[dr]^f & \\
Q' \ar[r]_{\gamma_a'} & Q' \ar[r]_{f'} & \two 
}
\]
In this case also the square below commutes, which implies that $h$ is a $T_\Sigma$-coalgebra homomorphism $h: Q_w\ra Q'_w$ for all $w\in\Sigma^*$.
\[
\xymatrix{
Q \ar[r]^{\gamma_w} \ar[d]_h & Q \ar[d]^h  \\
Q' \ar[r]_{\gamma_w'} & Q'
}
\]
\end{rem}

\begin{proposition}\label{prop:rqc2}
A subcoalgebra $Q$ of $\rho T_\Sigma$ is closed under right derivatives (i.e., $L\in Q$ implies $Lw^{-1}\in Q$ for each $w\in\Sigma^*$) iff there exists a coalgebra homomorphism from $Q_w$ to $Q$ for each $w\in\Sigma^*$.
\end{proposition}

\begin{proof}
A subcoalgebra $Q$ of $\rho T_\Sigma$ is up to isomorphism a set of regular languages over $\Sigma$ carrying a $\Cat$-algebraic structure and closed under left derivatives, see Corollary \ref{cor:ratlift}.
The ordering of subcoalgebras is induced by inclusion of sets, i.e.,\ $Q \leq Q'$ in $\Sub{\rho T_\Sigma}$ iff $Q \subseteq Q'$.

Suppose that a coalgebra homomorphism $\alpha_w : Q_w \to Q$ exists for each $w \in \Sigma^*$. The languages accepted by $Q_w$ are precisely $\{ L w^{-1} : L \in Q \}$ because we have moved the final states backwards along all $w$-paths. The morphism $\alpha_w$ assigns to each state of $Q_w$ its accepted language, so  $\{ L w^{-1} : L \in Q \} \subseteq Q$. Hence $Q$ is closed under right derivatives.

Conversely suppose that $Q \monoto \rho T_\Sigma$ is closed under right derivatives. Clearly $Q_w$ is locally finite for each $w\in \Sigma^*$, so there are unique homomorphisms $Q_w\ra \rho T_\Sigma$ and $\coprod_w Q_w \ra \rho T_\Sigma$. Factorize them as in Remark \ref{rem:algfact}(b):
\[
Q_w \epito {\widetilde{Q_w}} \monoto \rho T_\Sigma
\qquad\qquad
\coprod_{w \in \Sigma^*} Q_w \epito {\widetilde{Q}} \monoto \rho T_\Sigma
\]
It is now easy to see (using the factorization system) that
\[
{\widetilde{Q}} = \Lor \{ {\widetilde{Q_w}} : w \in \Sigma^* \}
\]
in $\Sub{\rho T_\Sigma}$. Since $Q$ is closed under right derivatives we have $\widetilde{Q_w}=\{ Lw^{-1} : L \in Q \} \subseteq Q$ and hence $\widetilde{Q} \subseteq Q$. Since $Q = Q_\epsilon$ the reverse inclusion holds, so $Q = \widetilde{Q}$. Therefore we obtain a coalgebra homomorphism $Q_w \xra{in_w} \coprod_{w \in \Sigma^*} Q_w \epito \widetilde{Q} = Q$.\qed
\end{proof}

\begin{lemma}
\label{lem:finsub_in_finrqc}
Every finite subcoalgebra of $\rho T_\Sigma$ is contained in a  finite subcoalgebra of $\rho T_\Sigma$ closed under right derivatives.
\end{lemma}

\begin{proof}
Let $Q \monoto \rho T_\Sigma$ be an finite subcoalgebra of $\rho T_\Sigma$. Since $\Cat(Q,Q)$ is finite, there exists a finite set $W\seq \Sigma^*$ of words such that for every $u \in \Sigma^*$ there exists $w \in W$ with $Q_u = Q_w$. Then the coalgebra $\coprod_{w \in W} Q_w$ is finite  because the coproduct is constructed on the level of $\Cat$, and $\Cat_f$ is closed under finite colimits as $\Cat$ is locally finite. Factorize the unique homomorphism  $\coprod_{w\in W} Q_w\ra \rho T_\Sigma$ as is Remark \ref{rem:algfact}(b):
\[
\xymatrix@=0pt{
\coprod_{w \in W} Q_w \ar@{>>}[rrr]^-e &&& \; {Q'} \; \ar@{>->}[rrr]^-m &&& \rho T_\Sigma\\
&&&&&&&\\
Q=Q_\epsilon \ar[uu]^{in_w} \ar@/_2ex/@{>-->}[uurrr] \ar@{>->}@/_3ex/[uurrrrrr] &&&&&&&
}
\]
The coalgebra ${Q'}$ is clearly also finite. Moreover, if $w$ is chosen such that $Q=Q_\epsilon=Q_w$, the above diagram yields $Q\monoto Q'$.

It remains to show that ${Q'}$ is closed under right derivatives. First observe that, for each $w \in \Sigma^*$, one has
\[
(\coprod_{w' \in W} Q_{w'})_w
= \coprod_{w' \in W} (Q_{w'})_w
= \coprod_{w' \in W} Q_{ww'}
\]
because $(-)_w$ commutes with coproducts and $f \circ \gamma_{w'} \circ \gamma_w = f \circ \gamma_{ww'}$. Consider the diagram of coalgebra homomorphisms
\[
\xymatrix{
\coprod_{w' \in W} Q_{ww'} = (\coprod_{w' \in W} Q_{w'})_w
\ar@{->>}[r]^-{e}
\ar[d]_-h
&
Q_w'
\ar@{->}[dd]^-\beta
\ar@{-->}[ddl]^{\alpha_w}
\\
\coprod_{w' \in W} Q_{w'}
\ar@{->>}[d]_e
\\
Q'
\ar@{>->}[r]_m
&
\rho T_\Sigma
}
\]
The strong epic $e$ and mono $m$ were defined above, $\beta$ is the final morphism, the topmost $e$ exists by Remark \ref{rem:rqcmor}, and $h$ exists because each $Q_{ww'}$ equals some $Q_v$ with $v \in W$. Hence the square commutes by finality, and diagonal fill-in gives a coalgebra homomorphism $\alpha_w: Q'_w\ra Q'$. Since $w$ was an arbitrary word we deduce from Proposition \ref{prop:rqc2} that $Q'$ is closed under right derivatives.
\qed
\end{proof}

\begin{cor}\label{cor:finsub_in_finrqc}
Let $Q\monoto Q' \monoto \rho T_\Sigma$ be subcoalgebras where $Q$ is finite and $Q'$ is closed under right derivatives. Then $Q$ is contained in a finite subcolgebra of $Q'$ closed under right derivatives.  
\end{cor}

\begin{proof}
By the previous lemma there exists a finite subcoalgebra $Q''\monoto \rho T_\Sigma$ that contains $Q$ and is closed under right derivatives. Since $T_\Sigma$ preserves intersections,  $Q''\cap Q'$ is a subcoalgebra of $Q'$ with the desired properties.
\end{proof}

\begin{notation}
$\Subrqc{\rho T_\Sigma}$ and $\FPSubrqc{\rho T_\Sigma}$ are the posets of (finite) subcoalgebras of $\rho T_\Sigma$ closed under right derivatives. 
\end{notation}

\begin{proposition}\label{prop:rqcidcomp}
$\Subrqc{\rho T_\Sigma}$ is the ideal completion of $\FPSubrqc{\rho T_\Sigma}$.
\end{proposition}

\begin{proof}
Since $T_\Sigma$ preserves intersections, $\Subrqc{\rho T_\Sigma}$ forms a complete lattice whose meet is set-theoretic intersection. It remains to check the conditions (1) and (2) of Remark \ref{rem:idcomp}.
For (2), just note that directed joins are directed unions of subcoalgebras, so that every finite subcoalgebra closed under right derivatives is compact. For (1)  we use that every subcoalgebra $Q \monoto \rho T_\Sigma$ is locally finite, see Proposition \ref{prop:locfin}, and hence arises as the directed union of its finite subcoalgebras. But by Corollary \ref{cor:finsub_in_finrqc} the poset of all finite subcoalgebras of $Q$ contains the ones closed unter right derivatives as a final subposet, so $Q$ is also the directed union of its  finite subcoalgebras closed under right derivatives.
\qed
\end{proof}

\subsection{$\boldsymbol{\DCat}$-Monoids} \label{sec:monoids}
By Proposition \ref{prop:rqc2} closure under right derivatives of a subcoalgebra $Q\monoto \rho T_\Sigma$ is characterized  by the existence of $T_\Sigma$-coalgebra homomorphisms $Q_w \ra Q$. In this section we investigate the dual $L_\Sigma$-algebra homomorphisms $\widehat{Q}\ra \widehat{Q_w}$  and show that they define a monoid structure on $\widehat{Q}$. This requires the following additional assumptions on the variety $\DCat$:

\begin{assumptions}\label{asm:geneilenberg} For the rest of Section \ref{sec:locvar} we assume that 
\begin{enumerate}[(a)]
\item epimorphisms in $\DCat$ are surjective; 
\item for any two algebras $A$ and $B$ in $\DCat$, the set $[A,B]$  of homomorphisms from $A$ to $B$  is an algebra in $\DCat$ with the pointwise algebraic structure, i.e., a subalgebra of the power $B^{\und{A}} = \prod_{x\in A} B$;
\item $\one = \widehat \two$ is a free $\DCat$-algebra on the one-element set $1$, that is, 
\[ \one \cong \Psi 1 \]
for the left adjoint $\Psi: \Set\ra \DCat$ to the forgetful functor $\DCat\ra\Set$.
\end{enumerate}
\end{assumptions}

\begin{rem}\label{rem:dmonoidal}
Here is a more categorical view of Assumption \ref{asm:geneilenberg}(b). Given algebras $A$, $B$ and $C$ in $\DCat$, a \emph{bimorphism} is a function $f: A\times B \ra C$ such that $f(a,\mathord{-}): B\ra C$ and $f(\mathord{-},b): A\ra C$ are $\DCat$-morphisms for all $a\in A$ and $b\in B$. A \emph{tensor product} of $A$ and $B$ is a universal bimorphism $t: A\times B \ra A\otimes B$, i.e., for every bimorphism $f: A\times B \ra C$ there is a unique $\DCat$-morphism $f'$ making the diagram below commute.
\[
\xymatrix{
A\times B \ar[dr]_f \ar[r]^t & A\otimes B \ar@{-->}[d]^{f'}\\
& C
}
\]
The tensor product exists in every variety $\DCat$ and turns   it into a symmetric monoidal category $(\DCat,\otimes,\Psi 1)$, see \cite{BN1976}. Assumption \ref{asm:geneilenberg} (b) then states precisely that $\DCat$ is monoidal closed. 
\end{rem}

\begin{expl}
The varieties $\DCat=\Set,\Poset,\JSL,\Vect{\Field_2}$ of Example \ref{expl:dualpairs} meet the Assumptions \ref{asm:geneilenberg}.
\end{expl}

\begin{rem}\label{rem:ev}
\begin{enumerate}[(1)]
\item For all algebras $A$ and $B$ in $\DCat$ and $x\in A$ let $\ev_x$ be the composite
\[  [A,B]\monoto B^{\und{A}} \xra{\pi_x} B.\] 
The morphism $\ev_x$ is evaluation at $x$, i.e., 
\[\ev_x(f) = f(x)\quad \text{ for all } f\in [A,B].\]
\item For any two $\DCat$-morphisms $f:B\ra B'$ and $g: A \ra A'$, the maps 
\[c_f: [A,B]\ra [A,B']\quad\text{and}\quad c_g': [A,B]\ra [A',B]\]
given by composition with $f$ and $g$, respectively, are $\DCat$-morphisms.
\end{enumerate}
\end{rem}
The assumption that $\DCat$ is monoidal closed gives rise to inductive definition and proof principles that we shall use extensively.

\begin{definition}[Inductive Extension Principle]\label{def:ext}
Let $(g_i: A\ra B)_{i\in I}$ be a set-indexed family of morphisms between two fixed $\DCat$-objects $A$ and $B$. Its \emph{inductive extension} is the family $(g_x: A\ra B)_{x\in \Psi I}$ defined as follows:
\begin{enumerate}[(1)]
\item Extend the function $g: I \ra [A,B]$, $i\mapsto g_i$, to a $\DCat$-morphism $\overline g: \Psi I \ra[A,B]$:
\[
\xymatrix{
I\ar[d]_{\eta} \ar[dr]^-{g} &\\
\Psi I \ar@{-->}[r]_<<<<<{\overline g} & [A,B]
}
\]
\item Put $g_x := \overline g(x)$ for all $x\in\Psi I$.
\end{enumerate}
\end{definition}

\begin{lemma}\label{lem:ext}
\begin{enumerate}[(a)]
\item In Definition \ref{def:ext} we have $g_{\eta i}=g_i$ for all $i\in I$. 
\item For all sets $I$ and $\DCat$-objects $A$, the family $(\ev_x : [\Psi I, A]\ra A)_{x\in \und{\Psi I}}$ is the inductive extension of $(\ev_{\eta i} : [\Psi I,A]\ra A)_{i\in I}$.
\end{enumerate}
\end{lemma}

\begin{proof}
\begin{enumerate}[(a)]
  \item For all $x\in A$ we have:
\begin{align*}
\und{g_{\eta i}}(x) &= \und{\ev_{x}}(g_{\eta i}) & \text{def. } & \ev\\
&= \und{\ev_{x}}(\und{\overline g}(\eta i)) & \text{def. } & g_{\eta i}\\
&= \und{\ev_{x}}(g_i) & \text{def. } & \ol g\\
&= \und{g_i}(x) & \text{def. } & \ev
\end{align*}
\item Extend the function $g: I\ra \und{[[\Psi I, A],A]}$, $i\mapsto \ev_{\eta i}$, to a $\DCat$-morphism $\bar g$:
\[
\xymatrix{
I\ar[d]_{\eta} \ar[dr]^-{g} &\\
\und{\Psi I} \ar[r]_-{\und{\overline g}} & \und{[[\Psi I,A],A]}
}
\]
We need to show that the extended family consists of $g_x = \ev_x$ for all $x\in\und{\Psi I}$. To this end, we first prove the equation
\[ \ev_f\cdot \ol g = f \quad\text{for all } f: \Psi I\ra A.\]
It suffices to prove $\und{\ev_f\cdot \ol g}\o\eta = \und{f}\o\eta$, and indeed we have for all $i\in I$:
\begin{align*}
\und{\ev_f}(\und{\ol g}(\eta i)) &= \und{\ev_f}(gi) & \text{def. } & \ol g\\
&= \und{\ev_f}(\ev_{\eta i}) & \text{def. } & g\\
&= \und{\ev_{\eta i}}(f) & \text{def. } & \ev\\
&= \und{f}(\eta i) & \text{def. } & \ev
\end{align*}
Therefore, for all $x\in\und{\Psi I}$ and $f: \Psi I\ra A$,
\begin{align*}
\und{g_x}(f) &= \und{\ev_f}(g_x) & \text{def. }  \ev\\
&= \und{\ev_f}(\und{\ol g}(x)) & \text{def. }  g_x\\
&= \und{f}(x) & \text{eqn. above}\\
&= \und{\ev_x}(f) & \text{def. }  \ev
\end{align*}
so $g_x=\ev_x$ as claimed. \qed
\end{enumerate}
\end{proof}

\begin{lemma}[Inductive Proof Principle]\label{lem:indproof}
Let $f,f'$, $h,h'$ and $g_i, g_i'$ ($i\in I$) be $\DCat$-morphisms as in the diagram $(\ast)$ below.
\[ 
\xymatrix@=10pt{
&  B \ar@{}[ddr]|{(\ast)} \ar[r]^{g_i} & C \ar[dr]^h & && &  B \ar@{}[ddr]|{(\ast\ast)} \ar[r]^{g_x} & C \ar[dr]^h & \\
A \ar[ur]^f \ar[dr]_{f'} & & & D && A \ar[ur]^f \ar[dr]_{f'} & & & D \\
& B' \ar[r]_{g_i'} & C' \ar[ur]_{h'} & && & B' \ar[r]_{g_x'} & C' \ar[ur]_{h'} & 
}
 \]
If $(\ast)$ commutes for all $i\in I$, then $(\ast\ast)$ commutes for all $x\in\Psi I$.
\end{lemma}

\begin{proof}
Suppose that $h\cdot g_i\cdot f = h'\cdot g_i' \cdot f'$ for all $i\in I$. We first prove the equation
\[ c_f'\cdot c_h \cdot \ol g = c_{f'}'\cdot c_{h'} \cdot \ol{g'} \]
where $c_h$, $c_{h'}$, $c_f'$ and $c_{f'}'$ are the $\DCat$-morphisms of Remark \ref{rem:ev}(b) and $\ol g, \ol {g'}: \Psi I \ra [B,C]$ are as in Definition \ref{def:ext}. Indeed, we have for all $i\in I$:
\begin{align*}
\und{c_f'}(\und{c_h}(\und{\ol g}(\eta i))) &= \und{c_f'}(\und{c_h}(g_i)) & \text{def. } \ol g\\
&= h\cdot g_i \cdot f & \text{def. } c_h, c_f'\\
&= h'\cdot g_i' \cdot f' & \text{assumption}\\
&= \und{c_{f'}'}(\und{c_{h'}}(\und{\ol{g'}}(\eta i))) & \text{compute backwards}
\end{align*}
We conclude that, for all $x\in\Psi I$,
\begin{align*}
h\cdot g_x \cdot f &= \und{c_f'}(\und{c_h}(g_x)) & \text{def. } c_h, c_f'\\
&=  \und{c_f'}(\und{c_h}(\und{\overline g}(x))) & \text{def. } g_x\\
&=  \und{c_{f'}'}(\und{c_{h'}}(\und{\overline {g'}}(x))) & \text{eqn. above}\\
&= h'\cdot g_x' \cdot f' & \text{compute backwards} \tag*{\qed}
\end{align*}
\end{proof}

\begin{definition}\label{def:bimonoid}
A \emph{$\DCat-$monoid} $(A,\circ,i)$ consists of an algebra $A$ in $\DCat$, a bimorphism $\circ: A\times A\ra A$ and an element $i: 1\ra A$ subject to the usual monoid axioms, i.e., the multiplication $\circ$ is associative and $i$ is its unit. A \emph{morphism} of $\DCat$-monoids
\[ h: (A,\circ,i) \ra (A',\circ',i') \]
  is a $\DCat$-morphism $h:A\ra A'$ that is also a monoid morphism, i.e, it preserves the multiplication and the unit. We denote by $\Bim$ the category of $\DCat$-monoids, and by $\Bimfp$ the full subcategory of all finite $\DCat$-monoids.
\end{definition}

\begin{rem}
$\DCat$-monoids are precisely the monoid objects in the monoidal category $(\DCat,\otimes,\Psi 1)$, see Remark \ref{rem:dmonoidal}.
\end{rem}

\begin{expl}
$\DCat$-monoids in $\DCat=\Set,\Poset,\JSL,\Vect{\Field_2}$ correspond to monoids, ordered monoids, idempotent semirings and $\Field_2$-algebras, respectively.
\end{expl}

\begin{lemma}\label{lem:bimprops}
\begin{enumerate}[(a)]
\item $\Bim$ is complete and $\Bimfp$ is finitely complete, with limits formed on the level of $\DCat$.
\item The (epi, strong mono)-factorization system of $\DCat$ lifts to $\Bim$.
\end{enumerate}
\end{lemma}

\begin{proof}
\begin{enumerate}[(a)]
\item Let $V:\Bim\ra\DCat$ denote the forgetful functor. Given a diagram $D: \mathcal{S}\ra \Bim$ where $D_s=(A_s,\circ_s, i_s)$ for $s\in\mathcal{S}$, form the limit cone $(L\xra{p_s} VD_s)_{s\in\mathsf{S}}$ in $\DCat$ (which is also a limit cone in $\Set$ as limits in $\DCat$ are formed one the level of underlying sets). Since \[(\und{L}\times\und{L} \xra{\und{p_s}\times\und{p_s}} \und{VD_s}\times \und{VD_s} \xra{\circ_s} \und{VD_s})_{s\in\mathcal{S}}\quad\text{and}\quad (1\xra{i_s} \und{VD_s})_{s\in\mathcal{S}}\]
are also cones (in $\Set$) over $\und{VD}$, there are unique mediating maps $\circ: \und{L}\times\und{L}\ra \und{L}$ and $i: 1\ra\und{L}$. It is now easy to verify that $(L,\circ,i)$ is $\DCat$-monoid and that $((L,\circ, i)\ra D_s)_{s\in\mathcal{S}}$ forms a limit cone over $D$ in $\Bim$. This proves the completeness of $\Bim$. The proof that $\Bimfp$ is finitely complete is identical -- just start with a finite diagram $D$ and use that $\DCat_{f}$ is closed under finite limits.
\item Let $f : (M,\bullet,i) \to (M',\bullet',i')$ be a morphism of $\DCat$-monoids and $M \xra{e} M_0 \xra{m} M'$ be its (epi, strong mono)-factorization in $\DCat$. Then $\und{m}$ is injective and $\und{e}$ is surjective by Assumption \ref{asm:geneilenberg}(a). Consequently, there exists a unique monoid structure $\star$ and $i_0$ on $\und{M_0}$ for which $\und{e}$ and $\und{m}$ are monoid morphisms. All that needs proving is that for every element $x \in \und{M_0}$ the function $x \star -$ is an endomorphism of $M_0$ in $\DCat$ (and similarly for $- \star x$). Let $d' : M' \to M'$ be the $\DCat$-morphism
\[
\und{d'} = x' \bullet' -
\qquad
\text{for $x' = \und{m}(x)$.}
\]
Since $\und{e}$ is surjective, we have $y \in \und{M}$ with $\und{e}(y) = x$ and we denote by $d : M \to M$ the $\DCat$-morphism $\und{d} = y \bullet -$. Then $f \cdot d = d' \cdot f$ because, for all $z\in M$,
\[
\begin{array}{rcl@{\qquad\qquad}p{5cm}}
  \und{f \cdot d}(z) & = & \und{f}(y \bullet z) & def.~$d$\\
  & = & \und{f}(y) \bullet' \und f (z) & $\und f$ monoid morphism \\
  & = & x' \bullet' \und f (z) & since $\und f(y) = x'$ \\
  & = & \und{d'}(\und f (z)) & def.~$d'$ \\
  & = & \und{d' \cdot f}(z).
\end{array}
\]
The unique diagonal fill in
\[
\xymatrix@=10pt{
M \ar[d]_d \ar[rr]^e && M_0 \ar[d]^m \ar@{-->}[ddll]_{d_0}
\\
M \ar[d]_e && M' \ar[d]^{d'}
\\
M_0 \ar[rr]_m && M'
}
\]
yields a $\DCat$-morphism $d_0$ with $\und{d_0} = x \star -$. Indeed, given $p \in \und{M_0}$, choose $q \in \und{M}$ such that $p = \und{e}(q)$. Then:
\[
\und{d_0}(p) = \und{d_0 \cdot e}(q) = \und{e} \cdot \und{d}(q) = \und{e}(y \bullet q) = \und{e}(y) \star \und{e}(q) = x \star p. \tag*{\qed}
\]
\end{enumerate}
\end{proof}

Our next goal is to show that the free monoid $(\Sigma^*,\o,\epsilon)$ on $\Sigma$ in $\Set$ extends to a free $\DCat$-monoid $(\Psi \Sigma^*, \bullet, \eta\epsilon)$ on $\Sigma$.

\begin{definition}	
For every word $w\in\Sigma^*$ the endomaps $w\cdot -$ and $-\cdot w$ of $\Sigma^*$ yield unique $\DCat$-endomorphisms $l_w$ and $r_w$ of $\Psi\Sigma^*$ making the squares below commute:
\[
\xymatrix{
\Sigma^* \ar[r]^{w \cdot -} \ar[d]_{\eta} & \Sigma^* \ar[d]^{\eta} & \Sigma^* \ar[r]^{-\cdot w} \ar[d]_{\eta} & \Sigma^* \ar[d]^{\eta}\\
\und{\PS} \ar[r]_{\und{l_w}} & \und{\PS} & \und{\PS} \ar[r]_{\und{r_w}} & \und{\PS} 
}
\]
Let $l_x, r_x: \PS\ra \PS$ ($x\in \und{\PS}$) be the arrows obtained by inductively extending the families $(l_w)_{w\in\Sigma^*}$ and $(r_w)_{w\in\Sigma^*}$, respectively.
\end{definition}

\begin{lemma}\label{lem:lrprops} For all $x,y\in \VPS$ the following equations hold:
\begin{enumerate}[(a)]
\item $r_x\cdot l_y = l_y \cdot r_x$
\item $\und{r_y}(x)=\und{l_x}(y)$
\end{enumerate}
\end{lemma}

\begin{proof}
\begin{enumerate}[(a)]
\item Consider the diagrams below:
\[
\xymatrix{
\PS \ar [r]^{r_w} \ar[d]_{l_v} & \PS \ar[d]^{l_v} & \PS \ar [r]^{r_x} \ar[d]_{l_v} & \PS \ar[d]^{l_v} & \PS \ar [r]^{r_x} \ar[d]_{l_y} & \PS \ar[d]^{l_y} & \\
\PS \ar[r]_{r_w} & \PS & \PS \ar[r]_{r_x} & \PS & \PS \ar[r]_{r_x} & \PS
}
\]
The left square commutes for all $v,w\in\Sigma^*$: indeed, for all $u\in\Sigma^*$ we have
\[ \und{l_v}(\und{r_w}(\eta u)) = \eta(vuw)  = \und{r_w}(\und{l_v}(\eta u)) \]
by the definition of $l_v$ and $r_w$. By induction it follows that the middle square commutes for all $v\in\Sigma^*$ and $x\in\VPS$. Using induction again we conclude that the right square commutes for all $x,y\in\VPS$.
\item We first prove that the following diagram commutes for all $y\in \VPS$ (where $\ol l$ is defined as shown in Definition~\ref{def:ext}(1)):
\[
\xymatrix{
\PS \ar[r]^-{\ol l} \ar@{=}[d] & [\Psi\Sigma^*,\Psi\Sigma^*] \ar[d]^{\ev_y}\\
\Psi\Sigma^* \ar[r]_-{r_y} & \Psi\Sigma^*
}
\]
By Lemma \ref{lem:ext}(b) and the induction principle it suffices to prove that it commutes for $y=\eta w$ where $w\in\Sigma^*$. In fact, we have for all $v\in\Sigma^*$:
\begin{align*}
\und{\ev_{\eta w}}(\und{\ol l}(\eta v)) &= \und{\ev_{\eta w}}(l_v) & \text{def. } \ol l\\
&= \und{l_v}(\eta w) & \text{def. } \ev\\
&= \eta(vw) & \text{def. } l_v\\
&= \und{r_w}(\eta v) & \text{def. } r_w\\
&= \und{r_{\eta w}}(\eta v) & \text{Lemma \ref{lem:ext}(a)}
\end{align*}
and therefore $\ev_{\eta w}\cdot \ol l = r_{\eta w}$. It follows that, for all $x,y\in\VPS$:
\begin{align*}
\und{r_y}(x) &= \und{\ev_y}(\und{\ol l}(x)) & \text{diagram above}\\
&= \und{\ev_y}(l_x) & \text{def. } l_x\\
&= \und{l_x}(y) &\text{def. } \ev \tag*{\qed}
\end{align*} 
\end{enumerate}
\end{proof}

\begin{definition}
We define a multiplication on $\VPS$ as follows:
\[ x\bullet y := \und{r_y}(x) = \und{l_x}(y) \quad\text{ for all } x,y\in \VPS.\]
\end{definition}

\begin{proposition}\label{prop:freebim}
$(\PS,\bullet,\eta\epsilon)$ is the free $\DCat$-monoid on $\Sigma$: for any $\DCat$-monoid $(A,\circ,i)$ and any function $f: \Sigma\ra \und{A}$, there is a unique extension to a $\DCat$-monoid morphism $\overline f: \PS\ra A$.
\[
\xymatrix{
\Sigma^* \ar[r]^-\eta & \VPS \ar@{-->}[d]^{\und{\overline f}}\\
\Sigma \ar@{ >->}[u] \ar[r]_f  & \und{A}
}
\]
\end{proposition}

\begin{proof}
We first show that $(\PS,\bullet,\eta\epsilon)$ is a $\DCat$-monoid. Indeed:
\begin{enumerate}[(1)]
\item $\bullet$ is associative: for all $x,y,z\in\VPS$ we have
\[x\bullet (y\bullet z) = \und{l_x}(\und{r_z}(y)) = \und{r_z}(\und{l_x}(y)) = (x\bullet y)\bullet z\]
using the definition of $\bullet$ and Lemma \ref{lem:lrprops}(a).
\item $\eta \epsilon$ is the neutral element: for all $x\in\VPS$ we have
\[x\bullet \eta \epsilon = \und{r_{\eta\epsilon}}(x) =\und{r_\epsilon}(x) = \und{\id}(x) = x\]
and symmetrically $\eta\epsilon \bullet x = x$.
\item $\bullet$ is a $\DCat$-bimorphism since, for all $x\in\VPS$, the functions $l_x = x\bullet -$ and $r_x = -\bullet x$ are $\DCat$-morphisms.
\end{enumerate}
It remains to verify the universal property. Given $f$ as in the diagram above, one can first extend $f$ to a monoid morphism $f': \Sigma^*\ra \und{A}$ (using that $\Sigma^*$ is the free monoid on $\Sigma$) and then extend $f'$ to a $\DCat$-morphism $\ol f:\Psi\Sigma^*\ra A$ with $\und{\ol f}\cdot \eta = f'$ (by the universal property of $\eta$). We only need to verify that $\und{\ol f}$ is a monoid morphism. Firstly, $\und{\ol f}$ preserves the unit:
\begin{align*}
\und{\ol f}(\eta \epsilon) &= f'(\epsilon) & \text{def. } \ol f\\
&= i & f' \text{ monoid morphism}
\end{align*}
To prove that $\und{\ol f}$ also preserves the multiplication, consider the squares below  where $l_x', r_w': A\ra A$ are the $\DCat$-morphisms  $\und{l_x'}=\und{\ol f}(x)\circ -$ and $\und{r_w'}=-\circ \und{\ol f}(\eta w)$.
\[
\xymatrix{
\PS \ar[r]^{r_w} \ar[d]_{\ol f} & \PS \ar[d]^{\ol f} & \PS \ar[r]^{l_x} \ar[d]_{\ol f} & \PS \ar[d]^{\ol f} \\
A \ar[r]_{r_w'} & A & A \ar[r]_{l_x'} & A
}
\]
The left square commutes for all $w\in\Sigma^*$ because, for all $v\in\Sigma^*$,
\begin{align*}
 \und{\ol f}(\und{r_w}(\eta v)) &= \und{\ol f}(\eta(vw)) & \text{def. } r_w\\
&= f'(vw) & \text{def. } \ol f\\
&= f'(v)\circ f'(w) & f' \text{ monoid morphism}\\
&= \und{\ol f}(\eta v)\circ \und{\ol f}(\eta w) & \text{def. } \ol f\\
&= \und{r_{w}'}(\und{\ol f}(\eta v)) & \text{def. } r_w'
\end{align*}
Then also the right square commutes for all $x\in\VPS$ because, for all $w\in\Sigma^*$,
\begin{align*}
\und{\ol f}(\und{l_x}(\eta w)) &= \und{\ol f}(\und{r_w}(x)) & \text{Lemma \ref{lem:lrprops}(b)}\\
&= \und{r_w'}(\und{\ol f}(x)) & \text{left square}\\
&= \und{\ol f}(x)\circ \und{\ol f}(\eta w) & \text{def. } r_w'\\
&= \und{l_x'}(\und{\ol f}(\eta w)) & \text{def. } l_x' 
\end{align*}
We conclude that, for all $x,y\in\VPS$,
\begin{align*}
\und{\ol f}(x\bullet y) &= \und{\ol f}(\und{l_x}(y)) & \text{def. } \bullet\\
&= \und{l_x'}(\und{\ol f}(y)) & \text{right square}\\
&= \und{\ol f}(x)\circ \und{\ol f}(y) & \text{def. } l_x'\tag*{\qed}
\end{align*}
\end{proof}

\begin{expl}
\begin{enumerate}[(a)]
\item For $\DCat=\Set$ or $\Poset$ we have  $\PS=\Sigma^*$ (discretely ordered in the case $\DCat=\Poset$). The monoid multiplication is concatenation of words, and the unit is $\epsilon$.
\item For $\DCat=\JSL$ we have $\PS=\FPow \Sigma^*$, the semilattice of all finite languages over $\Sigma$ w.r.t. union. The monoid multiplication is concatenation of languages, and the unit is the language $\{\epsilon\}$.
\item For $\DCat=\Vect{\Field_2}$ we have $\PS=\Pow \Sigma^*$, the vector space of all finite languages over $\Sigma$ where vector addition is symmetric difference. The monoid multiplication is the $\Field_2$-weighted concatenation  $L\otimes L'$ of languages, (i.e., $L \otimes L'$ consists of all words $w$ having an odd number of factorizations $w=uu'$ with $u\in L$ and $u'\in L'$), and the unit is again $\{\epsilon\}$.
\end{enumerate}
\end{expl}

\begin{definition}\begin{enumerate}[(a)]
\item A \emph{$\Sigma$-generated $\DCat$-monoid} is a quotient of $\PS$, i.e., a $\DCat$-monoid morphism $e: \Psi\Sigma^*\epito A$ with $e$ epic in $\DCat$. A \emph{morphism} between two $\Sigma$-generated $\DCat$-monoids $e: \PS\epito A$ and $e': \PS\ra A'$ is a generator-preserving $\DCat$-monoid morphism $f: A\ra A'$, i.e., $f\cdot e=e'$.
\item We denote by $\SBim$ the poset of all $\Sigma$-generated $\DCat$-monoids under the usual quotient ordering (see Notation \ref{not:order}), and by $\SBimfp$ the subposet of all $\Sigma$-generated finite $\DCat$-monoids.
\end{enumerate}
\end{definition}

\begin{rem}\label{rem:subdirprod}
The poset $\SBim$ is a complete lattice -- the join of a family $e_i: \PS\epito A_i$ ($i\in I$) of $\Sigma$-generated $\DCat$-monoids is their \emph{subdirect product} $S$,  obtained by forming their product in $\Bim$ and the (strong epi, mono)-factorization of the morphism $\langle e_i\rangle$:
\[  
\xymatrix{
\PS \ar@{->>}[d] \ar[drr]^{\langle e_i\rangle} \ar@{->>}[rr]^{e_i} && A_i\\
 S \ar@{>->}[rr] && \prod A_i \ar[u]_{\pi_i}
}
\]
Indeed, this follows from the fact that $\Bim$ is complete and inherits the factorization system of $\DCat$, see Lemma \ref{lem:bimprops}. Analogously, since $\Bimfp$ is finitely complete, $\SBimfp$ is a join-semilattice, in fact a join-subsemilattice of $\SBim$.
\end{rem}
$\Sigma$-generated monoids are closely related to algebras for the functor $L_\Sigma$.

\begin{notation}\label{not:mon}
In analogy to Notation \ref{not:rqc} we represent $L_\Sigma$-algebras $\one + \coprod_\Sigma A \xra{\alpha} A$  as triples
\[ A = (A, \alpha_a: A\ra A, i: \one\to A).\]
For any $L_\Sigma$-algebra $A=(A,\alpha_a,i)$ and $w\in\Sigma^*$ we put
\[ A_w := (A,\alpha_a, \alpha_w\cdot i)\]
where $ \alpha_w = \alpha_{a_n}\cdot \cdots \cdot \alpha_{a_1}$ for  $w=a_1\cdots a_n\in\Sigma^*$.
\end{notation}

\begin{rem}\label{rem:algmor}
Dually to Remark  \ref{rem:rqcmor}, a $\DCat$-morphism $h: A\ra A'$  is an $L_\Sigma$-algebra homomorphism $h: (A,\alpha_a,i)\ra (A',\alpha_a',i')$ iff the following diagram commutes for all $a\in\Sigma$:
\[
\xymatrix{
\one \ar[r]^i \ar[dr]_{i'} & A \ar[r]^{\alpha_a} \ar[d]_h & A \ar[d]^h  \\
& A' \ar[r]_{\alpha_a'} & A'
}
\]
In this case also the square below commutes, which implies that $h$ is $L_\Sigma$-coalgebra homomorphism $h: A_w\ra A'_w$ for all $w\in\Sigma^*$.
\[
\xymatrix{
A \ar[r]^{\alpha_w} \ar[d]_h & A \ar[d]^h  \\
A' \ar[r]_{\alpha_w'} & A'
}
\]
\end{rem}

\begin{rem}\label{rem:bimtoalg}
Every $\Sigma$-generated $\DCat$-monoid $e:\Psi\Sigma^*\epito (A,\circ,i)$ induces an $ L_\Sigma$-algebra $\tl A=(A,\alpha_a,\overline i)$ where $\alpha_a = -\circ \und{e}(\eta a): A\ra A$ and $\overline i: \Psi 1 = \one\ra A$ is the free extension of $i:1\ra \und{A}$. 
 In particular, for $e=\id$ we obtain
 \[\tl{\Psi\Sigma^*}=(\Psi\Sigma^*,r_a,\overline{\eta \epsilon}).\]
 Clearly every generator-preserving morphism $f:A\ra B$ of $\Sigma$-generated $\DCat$-monoids is also an $L_\Sigma$-algebra homomorphism $f: \tl A\ra \tl B$.
\end{rem}

\begin{proposition}\label{prop:initalg}
$\tl{\Psi\Sigma^*}=\mu  L_\Sigma$.
\end{proposition}

\begin{proof}
Consider the functor $ L_\Sigma^0=1+\Id^\Sigma$ on $\Set$. Since $ L_\Sigma^0$ and $ L_\Sigma$ are finitary, their initial algebras arise as colimits of the respective initial chains:
\[
\xymatrix@=10pt{
\emptyset \ar[r] & 1 \ar[r] & 1 + \coprod_\Sigma 1 \ar[r] & \dots
& \qquad &
0 \ar[r] & \one \ar[r] & \one + \coprod_\Sigma \one \ar[r] & \dots
}
\]
Now $\Psi$ preserves colimits (being a left adjoint) and $\one = \Psi 1$, so $\Psi$ maps the initial chain of $ L_\Sigma^0$ to the initial chain of $ L_\Sigma$ and preserves the colimit of the left chain, which implies that $\Psi(\mu  L_\Sigma^0)=\mu L_\Sigma$. Since $\mu  L_\Sigma^0=(\Sigma^*,-\cdot a,\epsilon)$ we have  $\mu  L_\Sigma=(\Psi\Sigma^*, \Psi(-\cdot a), \ol{\eta\epsilon})$. Moreover $\Psi(-\cdot a)=r_a$ by the definition of $r_a$, so 
\[\mu  L_\Sigma=(\Psi\Sigma^*, r_a, \ol{\eta\epsilon}) = \tl{\Psi\Sigma^*}. \tag*{\qed}\]
\end{proof}

\begin{prop}\label{prop:monalgiso}
$\SBim$ is a sublattice of $\Quo(\mu L_\Sigma)$, and $\SBimfp$ is a subsemilattice of $\Quofp{\mu L_\Sigma}$.
\end{prop}

\begin{proof}
Remark \ref{rem:bimtoalg} and Proposition \ref{prop:initalg} show that every (finite) $\Sigma$-generated $\DCat$-monoid $e: \Psi\Sigma^* \epito A$ induces a (finite) quotient algebra $e: \mu L_\Sigma \epito \widetilde A$ of $\mu L_\Sigma$ carried by the same morphism $e$. We show that the map $\tl{(-)}: \SBim \ra \Quo(\mu L_\Sigma)$ and its restriction $\tl{(-)}: \SBimfp \ra \Quofp{\mu L_\Sigma}$ are order-embeddings. Clearly $\tl{(-)}$ is injective and monotone. 
 To show that $\tl{(-)}$ it is order-reflecting, consider a commutative diagram as below, where $e$ and $e'$ are  $\DCat$-monoid morphisms and $f: \widetilde{A} \ra \widetilde{A'}$ is an $L_\Sigma$-algebra homomorphism.
\[
\xymatrix{
&\mu L_\Sigma=\PS \ar@{->>}[dl]_e \ar@{->>}[dr]^{e'} &\\
A \ar[rr]_f && A' 
}
\]
 We need to show that $\und{f}$ is a monoid morphism. Let $x',y'\in \und{A}$ and choose $x,y\in\und{\PS}$ with $e(x)=x'$ and $e(y)=y'$, using that $e$ is surjective by Assumption \ref{asm:geneilenberg}(a). Then
\[ \und{f}(x'\circ y') = \und{f}(\und{e}(x\bullet y)) = \und{e'}(x\bullet y) = \und{e'}x\circ' \und{e'}y = \und{fe}(x)\circ' \und{fe}(y) = \und{f}x'\circ' \und{f}y'\]
and moreover $\und{f}$ preserves the unit because $f$ is an $ L_\Sigma$-algebra homomorphism.  \qed
\end{proof}
Hence $\SBim$ is isomorphic to a sublattice of  $\Quo(\mu L_\Sigma)$. The following proposition characterizes its elements in terms of $L_\Sigma$-algebra homomorphisms.

\begin{proposition}\label{prop:bimtoalg}
A quotient algebra $e: \mu L_\Sigma \epito (A,\alpha_a,i)$ of $\mu L_\Sigma$ is induced by a $\Sigma$-generated $\DCat$-monoid if and only if there exists an $L_\Sigma$-algebra homomorphism from $A$ to $A_w$ for each $w\in\Sigma^*$.
\end{proposition}

\begin{proof}
($\Ra$) Suppose that  $e: \mu L_\Sigma \epito (A,\alpha_a, i)$ is induced by some $\Sigma$-generated $\DCat$-monoid $e: \PS \epito (A,\circ,i')$, that is, $\alpha_a=-\circ \und{e}(\eta a)$ and $i=\overline{i'}$. For each $w\in\Sigma^*$ consider the $\DCat$-morphism 
\[\beta_w = \und{e}(\eta w)\circ - : A\ra A.\] We claim that $\beta_w$ is an $ L_\Sigma$-algebra homomorphism
\[ \beta_w: (A,\alpha_a, i) \ra  (A,\alpha_a,\alpha_w \cdot i),\]
that is, a homomorphism $\beta_w: A\ra A_w$. In fact, we have for each $x\in \und{A}$ and $a\in\Sigma$:
\begin{align*} 
\und{\beta_w}(\und{\alpha_a}(x)) &= \und{e}(\eta w)\circ(x\circ \und{e}(\eta a)) &\text{def. } \beta_w,~\alpha_a\\
&= (\und{e}(\eta w)\circ x)\circ \und{e}(\eta a) & \text{associativity}\\
&= \und{\alpha_a}(\und{\beta_w}(x)) & \text{def. } \beta_w,~\alpha_a
\end{align*}
so $\beta_w\cdot \alpha_a = \alpha_a\o\beta_w$. To prove preservation of initial states (i.e., $\beta_w\cdot i=\alpha_w\cdot i$) it suffices to prove $\und{\beta_w\cdot  i}\cdot \eta_1=\und{\alpha_w\cdot  i}\o\eta_1$ where $\eta_1: 1\ra \und{\Psi 1} = \und{\one}$ is the unit. We compute:
\begin{align*}
\und{\beta_w\cdot i}\cdot \eta_1(*)&= \und{\beta_w}(i') & i = \ol{i'}\\
&= \und{e}(\eta w)\circ i' & \text{def. } \beta_w\\
&= \und{e}(\eta w) & i' \text{ neutral element of } A\\
&= \und{e}(\eta \epsilon \bullet \eta w) & \eta \epsilon \text{ neutral element of } \PS\\
&= \und{e}(\und{r_w}(\eta \epsilon)) & \text{def. } \bullet\\
&= \und{\alpha_w}(\und{e}(\eta \epsilon))& e \text{ } L_\Sigma\text{-algebra homomorphism}\\
&= \und{\alpha_w}(i') & e \text{ monoid morphism}\\
&= \und{\alpha_w}(\und{ i}(\eta_1(*))) &  i = \overline{i'}\\
&= \und{\alpha_w\cdot  i}\o\eta_1(*)
\end{align*}
Hence $\beta_w: A\ra  A_w$ is an $L_\Sigma$-algebra homomorphism as claimed.

($\La$) Suppose that an $L_\Sigma$-algebra homomorphism $\beta_w: A\ra A_w$ is given for each $w\in\Sigma^*$, and let $(\beta_x: A\ra A)_{x\in\VPS}$ and  $(\alpha_x: A\ra A)_{x\in\VPS}$ be the inductive extensions of the families $(\beta_w: A\ra A)_{w\in\Sigma^*}$ and  $(\alpha_w: A\ra A)_{w\in\Sigma}$ of $\DCat$-morphisms.
\begin{enumerate}[(1)]
\item For all $x\in\VPS$, let $A_x:=(A,\alpha_a, \alpha_x\cdot i)$. We claim that $\beta_x: A\ra A_x$ is an $ L_\Sigma$-algebra homomorphism, which means that the following squares commute:
\[
\xymatrix{
A \ar[r]^{\beta_x} \ar[d]_{\alpha_a} & A \ar[d]^{\alpha_a} & \one \ar[r]^i \ar[d]_{i} & A \ar[d]^{\beta_x}\\
A \ar[r]_{\beta_x} & A & A \ar[r]_{\alpha_x} & A
}
\]
Indeed, they clearly commute if $x=\eta(w)$ for some $w\in\Sigma^*$ because $\beta_{\eta w}=\beta_w: A\ra A_w$ is an $ L_\Sigma$-algebra homomorphism, and therefore they commute for all $x$ by the induction principle (Lemma \ref{lem:indproof}).
\item  We prove the equation
\[
\und{e}(x\bullet y) = \und{\alpha_y}(\und{e}(x)) \quad \text{for all } x,y\in\VPS,
\]
where $\bullet$ is the multiplication of the free $\DCat$-monoid $\PS = \mu L_\Sigma$. Observe first that the following diagram commutes for all $y\in\PS$:
\[
\xymatrix{
\PS \ar[r]^{r_y} \ar[d]_{e} & \PS \ar[d]^{e} \\
A \ar[r]_{\alpha_y} & A
}
\]
In fact, it commutes if $y=\eta w$ for some $w\in\Sigma^*$ because $e$ is an $ L_\Sigma$-algebra homomorphism, so it commutes for all $y$ by induction. Therefore
\[ \und{e}(x\bullet y) = \und{e}(\und{r_y}(x)) = \und{\alpha_y}(\und{e}(x)).\]
\item We prove the equation
\[ \und{e}(x\bullet y) = \und{\beta_x}(\und{e}(y)) \quad \text{for all }x,y\in\VPS.\]
First note that $l_x: \mu  L_\Sigma \ra (\mu L_\Sigma)_x$ is an $ L_\Sigma$-algebra homomorphism: we have $l_x \cdot r_a = r_a \cdot l_x$ by Lemma~\ref{lem:lrprops}(a) and $l_x \cdot \ol{\eta\epsilon} = r_x \cdot \ol{\eta\epsilon}$ because
\[
\und{l_x \cdot \ol{\eta\eps}} \cdot \eta_1(*) = \und {l_x}(\eta\eps) = x \bullet \eta \eps = \eta \eps \bullet x = 
\und{r_x}(\eta\eps) = \und{rx \cdot \ol{\eta\eps}} \cdot \eta_1(*). 
\]

 Since also $\beta_x: A\ra A_x$ is an $ L_\Sigma$-algebra homomorphism by (1), the following diagram commutes by initiality of $\mlhs$:
\[
\xymatrix{
\mu L_\Sigma \ar[r]^-{l_x} \ar[d]_e & (\mu L_\Sigma)_x \ar[d]^e \\
A \ar[r]_{\beta_x} & A_x
}
\]
Therefore
\[ \und{e}(x\bullet y) = \und{e}(\und{l_x}(y)) = \und{\beta_x}(\und{e}(y)) .\]
\item We define the desired monoid structure on $\und{A}$. The unit is $i\cdot \eta_1(\ast)\in A$, and the multiplication is given as follows: for all $x',y'\in \und{A}$, choose $x,y\in \VPS$ with $x'=\und{e}(x)$ and $y'=\und{e}(y)$ (using that $\und{e}$ is surjective by Assumption \ref{asm:geneilenberg}(a)) and put
\[ x'\circ y' := \und{e}(x\bullet y).\]
We need to prove that $x'\circ y'$ is well-defined, i.e., independent of the choice of $x$ and $y$. In fact, by (2) above, $x'\circ y'$ is independent of the choice of $x$ and moreover $- \circ y' = \alpha_y$. Analogously (3) states that $x'\circ y'$ is independent of the choice of $y$ and that $x' \circ - = \beta_x$. It follows that $\circ: A\times A \ra A$ is a well-defined $\DCat$-bimorphism, and by definition we have $e(x\bullet y) = e(x)\circ e(y)$ for all $x,y\in \PS$. The associative and unit laws for $\circ$ hold in $A$ because they hold in in $\PS$ and $\und{e}$ is surjective, concluding the proof that $(A,\circ, i\cdot \eta_1(\ast))$ is a $\DCat$-monoid. Moreover, $\und{e}$ is clearly a monoid morphism \[e: \PS\epito (A,\circ,i\cdot \eta_1(\ast)),\] and the quotient algebra of $\mu L_\Sigma$ it induces is precisely $(A,\alpha_a,i)$. For the latter we need to show $- \circ e(\eta a) = \alpha_a$ for all $a\in \Sigma$. Given $x'\in A$, we choose $x\in \PS$ with $e(x)=x'$ and compute 
\[ x' \circ e(\eta a) = e(x\bullet \eta a) = e(r_a(x)) = \alpha_a(e(x)) = \alpha_a(x'), \]
using the definitions of $\circ$ and $\bullet$ and the fact the $e$ is an $L_\Sigma$-algebra homomorphism.\qed
\end{enumerate}
\end{proof}

\subsection{The Local Eilenberg Theorem} \label{sec:proofloceil}

We can now put our (co-)algebraic characterizations of right derivatives and monoids together to prove the general local Eilenberg theorem. The key result is

\begin{prop}\label{prop:loceilfin}
The semilattices $\FPSubrqc{\rho T_\Sigma}$ and $\SBimfp$ are isomorphic.
\end{prop}

\begin{proof}
We show that the isomorphism $\FPSub{\rho T_\Sigma}\cong \Quofp{\rho T_\Sigma}$ of Proposition \ref{prop:subquo} restricts to an isomorphism  $\FPSubrqc{\rho T_\Sigma} \cong \SBimfp$. Indeed, a finite subcoalgebra $(Q,\gamma_a,f)\monoto \rho T_\Sigma$ is closed under right derivatives iff a $T_\Sigma$ coalgebra homomorphism $Q_w \ra Q$ exists for each $w\in \Sigma^*$ (see Proposition \ref{prop:rqc2}). By  Remark \ref{rem:algcoalgdual} the $L_\Sigma$-algebra dual to $Q=(Q,\gamma_a,f)$ is $\widehat{Q}=(\widehat{Q},\widehat{\gamma_a},\widehat{f})$, and the one dual to $Q_w = (Q,\gamma_a,f\cdot \gamma_w)$  is \[\widehat{Q_w}= (\widehat{Q}, \widehat{\gamma_a}, \widehat{\gamma_{w}}\cdot \widehat{f}) = \widehat{Q}_{w^r}\]where $w^r$ is the reversed word of $w$. Indeed, for $w=a_1\ldots a_n$ we have $\widehat{\gamma_w} = \widehat{\gamma_{a_1}} \cdot \cdots \cdot \widehat{\gamma_{a_n}}$. Hence by duality  an $L_\Sigma$-algebra homomorphism $\widehat{Q} \ra \widehat{Q}_{w^r}$ exists for each $w\in\Sigma^*$, which by Proposition \ref{prop:monalgiso} and \ref{prop:bimtoalg}  means precisely that $\widehat{Q} \in \SBimfp$.\qed
\end{proof}

\begin{definition}
A \emph{local variety of regular languages in $\Cat$} is a subcoalgebra of $\rho T_\Sigma$ closed under right derivatives. 
\end{definition}

\begin{expl}
\begin{enumerate}[(a)]
\item A \emph{local variety of regular languages} is a set of regular languages over $\Sigma$ closed under the boolean operations and derivatives and containing $\emptyset$. This is the case $\Cat=\BA$.
\item A \emph{local lattice variety of regular languages} is a set of regular languages over $\Sigma$ closed under union, intersection and derivatives and containing $\emptyset$ and $\Sigma^*$. This is the case $\Cat=\DL$.
\item A \emph{local semilattice variety of regular languages} is a set of regular languages over $\Sigma$ closed under union and derivatives and containing $\emptyset$. This is the case $\Cat=\JSL$.
\item A \emph{local linear variety of regular languages} is a set of regular languages over $\Sigma$ closed under symmetric difference and derivatives and containing $\emptyset$. This is the case $\Cat=\Vect{\Field_2}$.
\end{enumerate}
\end{expl}

\begin{definition}
A \emph{local pseudovariety of $\DCat$-monoids} is a set of finite $\Sigma$-generated $\DCat$-monoids closed under subdirect products and quotients, i.e., an ideal in the join-semilattice $\SBimfp$.
\end{definition}

This leads to the main result of this paper. For convenience, we recall all assumptions used so far in the statement of the theorem.

\begin{theorem}[General Local Eilenberg Theorem]\label{thm:loceil}
 Let $\Cat$ and $\DCat$ be predual locally finite varieties of algebras, where the algebras in $\DCat$ are possibly ordered. Suppose further that  $\DCat$ is monoidal closed w.r.t. tensor product, epimorphisms in $\DCat$ are surjective, and the free algebra in $\DCat$ on one generator is dual to a two-element algebra in $\Cat$. Then there is a lattice isomorphism
\[\text{local varieties of regular languages in $\Cat$} ~\cong \text{
 local pseudovarieties of $\DCat$-monoids.}\]

\end{theorem}

\begin{proof}
By Proposition \ref{prop:loceilfin} we have a semilattice isomorphism  
\[\FPSubrqc{\rho T_\Sigma} \cong \SBimfp\] 
Taking ideal completions on both sides yields a complete lattice isomorphism
\[\Ideal{\FPSubrqc{\rho T_\Sigma}} \cong \Ideal{\SBimfp}\] 
The ideals of the join-semilattice $\SBimfp$ are by definition precisely the pseudovarieties of $\DCat$-monoids. Moreover, Proposition \ref{prop:rqcidcomp} yields
\[\Ideal{\FPSubrqc{\rho T_\Sigma}} \cong \Subrqc{\rho T_\Sigma}, \]
 and the elements of $\Subrqc{\rho T_\Sigma}$ are by definition precisely the local varieties of regular languages in $\Cat$.\qed
\end{proof}

\begin{cor} By instantiating Theorem \ref{thm:loceil} to the categories of Example \ref{expl:dualpairs} we obtain the following lattice isomorphisms:
\end{cor}
\begin{center}
\begin{tabular}{|lllll|}
\hline\rule[11pt]{0pt}{0pt}
$\Cat$ & $\DCat$ & local varieties of regular languages &$\cong$& local pseudovarieties of \dots \\
\hline
$\BA$&$\Set$ &local varieties &$\cong$&  monoids\\
$\DL$&$\Poset$ &local lattice varieties &$\cong$&  ordered monoids\\
$\JSL$&$\JSL$ &local semilattice varieties  &$\cong$& idempotent semirings\\
$\Vect{\Field_2}$&$\Vect{\Field_2}$ & local linear varieties &$\cong$& $\Field_2$-algebras\\
\hline
\end{tabular}
\end{center}
The first two isomorphisms were proved in \cite{ggp08, ggp10}, the last two are new local Eilenberg correspondences.

\section{Profinite Monoids}
\label{sec:dual}
In \cite{ggp08,ggp10} Gehrke, Grigorieff and Pin demonstrated that the boolean algebra $\Reg_\Sigma$, equipped with left and right derivatives, is dual to the free $\Sigma$-generated profinite monoid. In this section we generalize this result to our categorical setting (Assumptions \ref{asm:sec3}). For this purpose we will introduce below a category $\DCath$ that is dually equivalent (rather than just predual) to $\Cat$, and arises as a ``profinite'' completion of $\DCat_f$.

\begin{defn} \begin{enumerate}[(a)]
\item Dually to Definition \ref{def:lfp}, an object $X$ of a category $\BCat$ is called \emph{cofinitely presentable} if the hom-functor $\BCat(-,X):\BCat\ra\Set^{op}$ is cofinitary, i.e., preserves cofiltered limits. The full subcategory of all cofinitely presentable objects is denoted by $\BCat_{cfp}$. The category $\BCat$ is \emph{locally cofinitely presentable} if $\BCat_{cfp}$ is essentially small, $\BCat$ is complete and every object arises as a cofiltered limit of cofinitely presentable objects.
\item The dual of $\mathsf{Ind}$ is denoted by $\mathsf{Pro}$:  the \emph{free completion under cofiltered limits} of a small category $\ACat$ is a full embedding $\ACat\hookrightarrow \Pro{\ACat}$ such that $\Pro{\ACat}$ has cofiltered limits and every functor $F: \ACat \ra \BCat$ into a category $\BCat$ with cofiltered colimits has an essentially unique cofinitary extension $\overline{F}: \Pro{\ACat}\ra \BCat$:
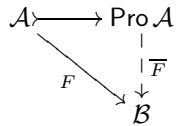
\[
\xymatrix{
\ACat \ar@{>->}[r] \ar[dr]_F & \Pro{\ACat} \ar@{-->}[d]^{\overline{F}} \\
& \BCat
}
\]
Note that $(\Pro \ACat)^{op} \cong \IndC{\ACat^{op}}$.  If $\ACat$ has finite limits then $\Pro{\ACat}$ is locally cofinitely presentable and $(\Pro{\ACat})_{cfp} \cong \ACat$. Conversely, for every locally cofinitely presentable category $\BCat$ one has $\BCat\cong\ProC{\BCat_{cfp}}$.
\end{enumerate}
\end{defn}

\begin{notation}
Let $\DCath$ denote the free completion of $\DCat_f$ under cofiltered limits, i.e.,
\[ \DCath = \ProC{\DCat_f}.\]
\end{notation}

\begin{rem}
$\DCath$ is dually equivalent to $\Cat$ since
\[ \Cat = \IndC{\Cat_f} \cong \IndC{\DCat_f^{op}} \cong \ProC{\DCat_f}^{op} =\DCath^{op}.\]
Moreover, the equivalence functors
\[ \widehat{(\mathord{-})} : \Cat_f \xra{\cong} \DCat_f^{op} \quad\text{and}\quad \overline{(\mathord{-})}: \DCat_f^{op} \xra{\cong} \Cat_f\]
extend essentially uniquely to equivalences -- also called $\widehat{(\mathord{-})}$ and $\overline{(\mathord{-})}$ -- between $\Cat$ and $\DCath^{op}$:
\[
\xymatrix{
\Cat = \IndC{\Cat_{f}} \ar[r]^<<<<<{\widehat{(\mathord{-})}} & \DCath^{op} && \Cat  & \ProC{\DCat_f}^{op} = \DCath^{op} \ar[l]_>>>>>>{\overline{(\mathord{-})}}\\
\Cat_f \ar@{>->}[u] \ar[r]_{\widehat{(\mathord{-})}}  & \DCat_f^{op} \ar@{>->}[u] && \Cat_f \ar@{>->}[u] & \DCat_f^{op} \ar@{>->}[u] \ar[l]^{\overline{(\mathord{-})}}
}
\]
\end{rem}

\begin{expl}\label{expl:dualpairs2}
For the predual varieties $\Cat$ and $\DCat$ of Example \ref{expl:dualpairs} we have the following categories $\DCath$:
\begin{center}
\begin{tabular}{|lll|}
\hline\rule[11pt]{0pt}{0pt}
$\Cat\quad\quad\quad$ & $\DCat\quad\quad\quad$  & $\DCath\quad\quad\quad$ \\
\hline
$\BA$ & $\Set$ & $\Stone$\\ 
$\DL$ & $\Poset$ & $\Priest$ \\ 
$\JSL$ & $\JSL$ & $\JSL$ in $\Stone$\\
$\Vect{\Field_2}$ & $\Vect{\Field_2}$ & $\Vect{\Field_2}$ in $\Stone$\\
\hline
\end{tabular}
\end{center}
In more detail:
\begin{enumerate}[(a)]
  \item For $\Cat=\BA$ and $\DCat=\Set$, we have the classical Stone duality: $\DCath$ is the category $\Stone$ of Stone spaces (i.e., compact Hausdorff spaces with a base of clopen sets) and continuous maps. 
The equivalence functor $\overline{(\mathord{-})}: \Stone^{op} \ra\BA$ assigns to each Stone space the boolean algebra of clopen sets, and its associated equivalence $\widehat{(\mathord{-})}: \BA \ra \Stone^{op}$ assigns to each boolean algebra the Stone space of all ultrafilters.
\item For the category $\Cat=\DL$ and $\DCat=\Poset$ we have the classical Priestley duality: $\DCath$ is the category $\Priest$ of Priestley spaces (i.e., ordered Stone spaces such that given $x\not\leq y$ there is a clopen set containing $x$ but not $y$) and continuous monotone maps. The equivalence functor $\overline{(\mathord{-})}: \Priest^{op}\ra\DL$ assigns to each Priestley space the lattice of all clopen upsets, and its associated equivalence $\widehat{(\mathord{-})}: \DL \ra \Priest^{op}$ assigns to each distributive lattice the Priestley space of all prime filters.
\item For $\Cat=\DCat=\JSL$ the dual category $\DCath$ is the category of join-semilattices in $\Stone$.
Similarly, for $\Cat=\DCat=\Vect{\Field_2}$ the dual category $\DCath$ is the category of $\Field_2$-vector spaces in $\Stone$, see \cite{j82}.
\end{enumerate}
\end{expl}

\begin{defn}
We denote by $\hL:\DCath\ra\DCath$ the dual of the functor $T:\Cat\ra\Cat$, i.e., the essentially unique functor for which the following diagram commutes up to natural isomorphism:
\[
\xymatrix{
\DCath^{op} \ar[r]^{\hL^{op}} \ar@{}[dr]|{\cong} & \DCath^{op} \\
\Cat  \ar[u]^{\widehat{(\mathord{-})}} \ar[r]_{T} & \Cat   \ar[u]_{\widehat{(\mathord{-})}} 
}
\]
\end{defn}

\begin{rem}\label{rem:algcoalgdual2}
In analogy to Remark \ref{rem:algcoalgdual}, the categories $\Coalg{T}$ and $\Alg{\hL}$ are dually equivalent: the equivalence $\widehat{(\mathord{-})}: \Cat \xra{\cong} \DCat^{op}$ lifts to an equivalence $\Coalg{T} \xra{\cong} (\Alg{\hL})^{op}$ given by 
\[ (Q\xra{\gamma} TQ) \quad\mapsto\quad  (\hL\widehat{Q} = \widehat{TQ} \xra{\widehat \gamma} \widehat Q). \]
\end{rem}

\begin{expl}\label{exp:dualfunc}
The dual of the endofunctor $T_\Sigma Q = \two \times Q^\Sigma$ of $\Cat$, see Example \ref{exp:tcoalg}, is the endofunctor of $\DCat$
  \[ \hL_\Sigma Z = \one + \coprod_\Sigma Z\]
 where $\one=\widehat{\two}$. In $\DCath=\Stone$ the object $\one$ is the one-element space. Hence, by the universal property of the coproduct,  an $\hL_\Sigma$-algebra $\hL_\Sigma Z = \one + \coprod_\Sigma Z \ra Z$ is a deterministic $\Sigma$-automaton (without final states) in $\Stone$, given by a Stone space $Z$ of states, continuous transition functions $\alpha_a: Z\ra Z$ for $a\in\Sigma$, and an initial state $\one\ra Z$. Analogously for the other dualities of Example \ref{expl:dualpairs2}.
\end{expl}

\begin{notation}
The category of all $\hL$-algebras with a cofinitely presentable carrier (shortly \emph{cfp-algebras}) is denoted by $\Algcfp{\hL}$.
\end{notation}

\begin{rem}
Note that $\Algcfp{\hL} \cong \FAlg{L}$ because the restrictions of $\hL$ and $L$  to $\DCath_{cfp} \cong \DCat_f$ are naturally isomorphic.
\end{rem}

\begin{defn}
Dually to Definition \ref{rem:rat_fix} an $\hL$-algebra is called \emph{locally cofinitely presentable} if it is a cofiltered limit of cfp-algebras.
\end{defn}

\begin{rem}\label{rem:lcpalg}
The category of all \lcp algebras is equivalent to
$\ProC{\Algcfp{\hL}}$.
This is the dual of Theorem \ref{thm:lfpindcomp}. The initial object $\tau \hL$ is what one can call the \emph{dual of the rational fixpoint}. By the dual of Definition \ref{rem:rat_fix}, one can construct $\tau \hL$ as the limit of all algebras in $\Alg{\hL}$ with cofinitely presentable carrier.
\end{rem}

\begin{expl}
\begin{enumerate}[(a)]
  \item For $\Cat=\BA$ and $\DCath=\Stone$, we have
  \[\tau \hL_\Sigma =  \text{ultrafilters of regular languages.}\]
  \item Analogously, for $\Cat=\DL$ and $\DCath=\Priest$, we have \[\tau \hL_\Sigma = \text{prime filters of regular languages.}\]
\end{enumerate}
\end{expl}

\begin{defn}
\label{not:F}
We denote by $F:\DCat\ra\DCath$ the essentially unique finitary functor for which
\[
\xymatrix@=5pt{
& \ar@{_{(}->}[dl]\DCat_f \ar@{^{(}->}[dr]& \\
\DCat=\IndC{\DCat_{f}}\ar[rr]_F & & \ProC{\DCat_{f}}=\DCath
}
\]
commutes, and by $U:\DCath\ra\DCat$ the essentially unique cofinitary functor for which
\[
\xymatrix@=5pt{
& \ar@{_{(}->}[dl]\DCat_{f} \ar@{^{(}->}[dr]& \\
\DCath=\ProC{\DCat_f}\ar[rr]_U & & \IndC{\DCat_f}=\DCat
}
\]
commutes.
\end{defn}

\begin{lem}\label{lem:ufadj}
The functors $F$ and $U$ are well-defined and $F$ is a left adjoint to $U$.
\end{lem}

\begin{proof}
$F$ is well-defined because $\DCath$ has filtered colimits: being \lcp it is cocomplete. Analogously for $U$. Furthermore, $F$ and $U$ form an adjunction: every object $A\in\DCat$ is a filtered colimit
\[ A=\colim_{i\in I} A_i \text{ with } A_i\in\DCat_{f}.\]
This implies, since $FA_i=A_i$, that
\[ FA = \colim_{i\in I} A_i \text{ in } \DCath.  \]
Analogously, every object $B\in \DCath$ is a cofiltered limit
\[ B=\lim_{j\in J} B_j \text{ with } B_j\in \DCat_{f}. \]
This implies, since $UB_j=B_j$, that,
\[ UB = \lim_{j\in J} B_j \text{ in } \DCat.\]
Consequently we have the desired natural isomorphism
\[ \DCath(FA,B) \cong \lim_{i\in I}\lim_{j\in J} \DCat_{fp}(A_i, B_j) \cong \DCat(A, UB). \tag*{\qed}\]
\end{proof}

\begin{expl}
\begin{enumerate}
\item
For $\Cat = \BA$, $\DCat = \Set$ and $\DCath = \Stone$, the functor $F : \Set \to \Stone$ is the Stone-\v{C}ech compactification and $U : \Stone \to \Set$ is the forgetful functor.
\item
For $\Cat = \DL$, $\DCat = \Poset$ and  $\DCath = \Priest$, the functor $F : \Poset \to \Priest$ constructs the free Priestley space on a poset and $U : \Priest \to \Poset$ is the forgetful functor.
\end{enumerate}
\end{expl}


\begin{defn}
Let $\hU: \ProC{\Algcfp{\hL}}\ra\Alg{L}$ denote the essentially unique cofinitary functor that makes the triangle below commute:
\[
\xymatrix@=5pt{
& \ar@{ )->}[dl]\Algcfp{\hL}\cong\Algfp{L} \ar@{ (->}[dr]& \\
\ProC{\Algcfp{\hL}}\ar[rr]_(.55){\hU} & & \Alg{L}
}
\]
\end{defn}

\begin{expl}
For $T_\Sigma Q = \two \times Q^\Sigma : \BA \to \BA$ we have $\hL_\Sigma Z
= \one + \coprod_\Sigma Z: \Stone \to \Stone$ and $L_\Sigma A = \one
+ \coprod_\Sigma A : \Set \to \Set$. The objects of
$\ProC{\Algcfp{\hL_\Sigma}}$ are the \lcp $\hL_\Sigma$-algebras, and the
functor $\hU : \ProC{\Algcfp{\hL_\Sigma}} \to \Alg{L_\Sigma}$
simply forgets the topology on the carrier of an $\hL_\Sigma$-algebra. 
\end{expl}

\begin{prop}\label{prop:hu_adj}
$\hU$ is a right adjoint.
\end{prop}

\begin{proof}
We have a commutative square
\[
\xymatrix@=5pt{
\ProC{\Algcfp{\hL}} \ar[dd] \ar[rr]^(.65){\hU} && \Alg{ L} \ar[dd]\\
&&\\
\Pro \DCath_{cfp} = \DCath \ar[rr]_(.7)U && \DCat
}
\]
where the vertical functors are the obvious forgetful functors. Recall that limits in $\Alg{ L}$ are formed on the level of $\DCat$, analogously for limits in $\ProC{\Algcfp{\hL}}$. Since $U$ preserves limits by Lemma \ref{lem:ufadj}, we conclude that so does $\hU$. By the Special Adjoint Functor Theorem, $\hU$ is a right adjoint: the category $\ProC{\Algcfp{\hL}}$ is complete, $\Algcfp{\hL}$ is its (essentially small) cogenerator, and $\ProC{\Algcfp{\hL}}$ is wellpowered. Indeed, it is dual to a \lfp category which is co-wellpowered by \cite[Theorem 1.58]{ar94}.\qed
\end{proof}

\begin{rem}
It follows that the left adjoint $\hF$ of $\hU$ maps the initial $L$-algebra to the initial \lcp $\hL$-algebra: $\hF(\mu L) = \tau \hL$. 
One can prove that $\hF$ assigns to
every $L$-algebra $\alpha: LA \ra A$ the limit of the diagram of all
its quotients in $\Algfp{L}=\Algcfp{\hL}$. Thus, we see that $\tau \hL$
can be constructed as the limit (taken in $\Alg{\hL}$) of all finite quotient $L$-algebras of $\mu L$. This construction generalizes a similar one given by Gehrke \cite{gehrke13}.
\end{rem} 

\begin{rem} Under the Assumptions \ref{asm:geneilenberg} of the previous section we obtain a generalization of the result of Gehrke, Grigorieff and Pin \cite{ggp08,ggp10} that $\Reg_\Sigma$ endowed with boolean operations and derivatives is dual to the free profinite monoid on $\Sigma$.
By Proposition \ref{prop:loceilfin} and Lemma \ref{lem:finsub_in_finrqc}  the finite $\Sigma$-generated $\DCat$-monoids form a cofinal subposet of $\Quofp{\mu L_\Sigma}$. Thus, the corresponding diagrams have the same limit in $\AlgCat{\hL_\Sigma}$. Hence by the previous remark $\tau \hL_\Sigma$ is also the limit of the directed diagram of all finite $\Sigma$-generated $\DCat$-monoids. Since $\hat{U}$ preserves this limit and the forgetful functor $\Mon{\DCat}$ creates limits (see Lemma \ref{lem:bimprops}) it follows that $\hat{U}(\tau \hL_\Sigma)$ carries the structure of a $\DCat$-monoid and it is then easy to see that it is the free profinite $\DCat$-monoid on $\Sigma$: for every $\DCat$-monoid morphism $e: \PS\ra M$ into a finite $\DCat$-monoid $M$, there exists a unique $\hL_\Sigma$-algebra homomorphism $\overline{e}: \tau \hL_\Sigma \ra M$ such that $\hat{U}\overline{e}$ is a $\DCat$-monoid morphism and the diagram below commutes:
\[
\xymatrix{
\PS \ar[r]^<<<<{\hat \eta} \ar@{->}[dr]_e & \hat U \hat F (\mu L_\Sigma) = \hat U (\tau \hL_\Sigma) \ar[d]^{\hat U\overline{e}}\\
& M = \hat U \hat F M
}
\] 
Here $\hat\eta$ is the unit of the adjunction $\hat F\dashv \hat U$. In summary:
\end{rem}

\begin{thm}
 Under the assumptions of the General Local Eilenberg Theorem, $\tau \hL_\Sigma$ is the free profinite $\DCat$-monoid on $\Sigma$.
\end{thm}

\section{A Categorical Framework}\label{sec:oldasm}

Although we have assumed $\Cat$ and $\DCat$  to be locally finite varieties throughout this paper, our methodology was purely categorical (rather than algebraic) in spirit. In fact,  all our results and their proofs can be adapted to the following categorical setting:
\begin{enumerate}
\item $\Cat$ and $\DCat$ are predual categories, i.e., the categories $\Cat_f$ and $\DCat_f$ of finitely presentable objects are dually equivalent.
\item $\Cat$ has the following additional properties:
\begin{enumerate}
\item  $\Cat$ is \lfp.
\item $\Cat_f$ is closed under strong epimorphisms.
\item $\Cat$ is concrete, i.e., a faithful functor $\under{\mathord{-}}_{\Cat}: \Cat\ra\Set$ is given. 
\item $\under{\mathord{-}}_{\Cat}$ is a finitary right adjoint and maps finitely presentable objects to finite sets.
\item  $\under{\mathord{-}}_{\Cat}$ is amnestic, i.e., every subset $B\seq \under{A}_{\Cat}$, where $A$ is an object of $\Cat$, is carried by at most one subobject of $A$ in $\Cat$.
\end{enumerate}
\item $\DCat$ has the following additional properties:
\begin{enumerate}
\item $\DCat$ is \lfp.
\item $\DCat_f$ is closed under strong monomorphisms and finite products.
\item $\DCat$ is concrete, i.e., a faithful functor $\under{\mathord{-}}_{\DCat}: \DCat\ra\Set$ is given. 
\item $\under{\mathord{-}}_{\DCat}$ is a finitary right adjoint and maps finitely presentable objects to finite sets.
\item  $\under{\mathord{-}}_{\DCat}$ preserves epimorphisms.
\item For any two objects $A$ and $B$ of $\DCat$, there exists an embedding $[A,B]\xto{e_{A,B}} B^{\under{A}}$ given by hom-sets, i.e.,
\begin{enumerate}
\item $[A,B]$ has the underlying set $\DCat(A,B)$,
\item $e_{A,B}$ is a monomorphism,
\item $\under{e_{A,B}}$ takes $f: A\ra B$ to its underlying function $|f|\in \under{B}^{\under{A}}$, and
\item whenever a $\DCat$-morphism $h: X\ra [A,B]$ factorizes through $e_{A,B}$ in $\Set$, it factorizes through $e_{A,B}$ in $\DCat$. 
\end{enumerate}
\end{enumerate}
\item The $\DCat$-object $\one = \Psi 1$, where $\Psi: \Set\ra\DCat$ is the left adjoint of $\under{\mathord{-}}_{\DCat}$, is dual to an object $\two$ of $\Cat$ with two-element underlying set. 
\item $T: \Cat \ra \Cat$ and $L: \DCat\ra \DCat$ are finitary predual functors, i.e., they restrict to functors $T_f: \Cat_f\ra\Cat_f$ and $L_f: \DCat_f\ra\DCat_f$ and these restrictions are dual. Moreover, $T$ preserves monomorphisms and preimages, and $L$ preserves epimorphisms. In Section \ref{sec:locvar} one works with the functors $T= T_\Sigma = \two\times\Id^\Sigma$ on $\Cat$ and $L = L_\Sigma = \one + \coprod_\Sigma \Id$ on $\DCat$.
\end{enumerate}

\section{Conclusions and Future Work}

Inspired by recent work of Gehrke, Grigorieff and Pin \cite{ggp08, ggp10} we have proved a
generalized local Eilenberg theorem, parametric in a pair of dual categories $\Cat$ and $\DCath$ and a type of coalgebras $T:\Cat\ra\Cat$. By instantiating our framework to  deterministic automata, i.e., the functor $T_\Sigma=\two\times \Id^\Sigma$ on $\Cat=\BA$, $\DL$, $\JSL$ and $\Vect{\Field_2}$, we derived the local Eilenberg theorems for (ordered) monoids as in ~\cite{ggp08}, as well as two new  local Eilenberg theorems for idempotent semirings and $\Field_2$-algebras.

There remain a number of open points for further work. Firstly, our
general approach should be extended to the
ordinary (non-local) version of Eilenberg's theorem. Secondly, for different functors $T$ on the categories we have
considered our approach should provide the means to relate varieties of rational behaviours of $T$ with
varieties of appropriate algebras. In this way, we hope to obtain
Eilenberg theorems for systems such as Mealy and Moore automata, but also weighted or probabilistic automata -- ideally, such results would be proved
uniformly for a certain class of functors.

Another very interesting aspect we have not treated in this paper are profinite
equations and syntactic presentations of varieties (of $\DCat$-monoids or
regular languages, resp.) as in the work of Gehrke, Grigorieff and
Pin~\cite{ggp08}. An important role in studying profinite equations will be played by the $\hL$-algebra $\tau \hL$, the dual of the rational fixpoint, that we identified as the free profinite
$\DCat$-monoid. A profinite equation is then a pair of elements of $\tau \hL$. We
intend to investigate this in future work.

\bibliographystyle{splncs03}		
\bibliography{EilenbergThm_fossacs2014_final}

\end{document}